\documentclass[12pt,english]{article}
\usepackage[T1]{fontenc}
\usepackage[latin9]{inputenc}
\usepackage{geometry}
\usepackage{amsthm}
\usepackage{amsmath}
\usepackage{amssymb}
\usepackage{graphicx}
\usepackage{chngcntr}
\usepackage{apptools}
\usepackage{natbib}
\bibpunct{(}{)}{,}{a}{,}{,}
\makeatletter
\numberwithin{equation}{section}
\numberwithin{figure}{section}

\usepackage{amsfonts}
\usepackage{amsmath}
\usepackage{amsthm}
\usepackage{color}
\usepackage[doublespacing]{setspace}
\usepackage[bottom]{footmisc}
\usepackage{indentfirst}
\usepackage{endnotes}
\usepackage{graphicx}
\usepackage{rotating}
\usepackage{lscape}
\usepackage{longtable}
\usepackage{pdflscape}
\usepackage{floatpag}
\usepackage{array}
\usepackage{calligra}
\usepackage[unicode=true,pdfusetitle, bookmarks=true,bookmarksnumbered=false,bookmarksopen=false, breaklinks=false,pdfborder={0 0 1},colorlinks=false]{hyperref}
\hypersetup{citecolor=blue, colorlinks=true}
\usepackage[all]{xy}
\usepackage{pifont}

\newcommand\trailstab{trail-stable}
\newcommand\trailstabity{trail stability}

\usepackage[normalem]{ulem}
\DeclareMathAlphabet{\mathcalligra}{T1}{calligra}{m}{n}
\theoremstyle{definition}
\newtheorem{theorem}{Theorem}
\newtheorem{corollary}{Corollary}

\newtheorem{definition}{Definition} %I think this should be numbered with thms
\newtheorem{example}{Example}

\newtheorem{lemma}{Lemma} %I think this should be numbered with thms

\newtheorem{proposition}{Proposition}
\newtheorem{obs}[theorem]{Observation}  %NEW

\DeclareMathAlphabet{\mathcalligra}{T1}{calligra}{m}{n}
\AtAppendix{\counterwithin{lemma}{section}}
\AtAppendix{\counterwithin{definition}{section}}
\AtAppendix{\counterwithin{theorem}{section}}
\counterwithout{figure}{section}

\date{}
\makeatother

\usepackage{babel}
\def\wsm{\widetilde \setminus} 
\def\wcup{\,\widetilde\cup\,} 
\def\wcap{\,\widetilde\cap\,}

\long\def\ch#1{#1}            %changes are not indicated

\long\def\info#1{}            %changes are not indicated

\begin{document}
\singlespacing

%\title{Trading networks with bilateral contracts}
\title {Trading Networks with Bilateral Contracts\thanks{{
%\scriptsize This paper unites three independent works \textit{Trading networks with bilateral contracts} circulated by Teytelboym in early 2014 (a much earlier draft that appeared in his doctoral thesis in 2013), Jank\'o's Master's thesis \textit{Generalized stable matchings: theory and applications}, supervised by Fleiner and completed, but not widely circulated, in 2011, as well as \textit{Stability of generalized network flows} by Fleiner, Jank\'o, and Tamura, completed, but not widely circulated, in 2014. %We learned about each others' work during the 12th Meeting of the Society of Social Choice and Welfare (Chestnut Hill) on 18-21 June 2014 and started communicating intensively in late 2014. An abstract of this paper appeared in the proceedings of MATCH-UP 2015 and AMMA 2015.
We would like to thank Samson Alva, Ravi Jagadeesan, Scott Kominers, Michael Ostrovsky, and four anonymous referees for their valuable comments. Vincent Crawford, Umut Dur, Jens Gudmundsson, Claudia Herrestahl, Paul Klemperer, Collin Raymond, and Zaifu Yang also gave great comments on much earlier drafts. We had enlightening conversations with Alex Nichifor, Alex Westkamp, and M. Bumin Yenmez about the project. Moreover, we are also grateful to seminar participants at the Southern Methodist University, National University of Singapore, ECARES, CIREQ Matching Conference (Montr\'{e}al), Workshop on Coalitions and Networks (Montr\'{e}al), the 12th and the 13th Meetings of the Society of Social Choice and Welfare, the 3rd International Workshop on Matching Under Preferences (Glasgow) and AMMA (Chicago) for their comments.}}} 

\author{
Tam\'{a}s Fleiner\thanks{{Budapest University of
Technology and Economics, Budapest, Hungary. Research was supported by the OTKA K108383
    research project and the MTA-ELTE Egerv\'ary Research Group. Part
  of the research was carried out during two working visits at Keio
  University. E-mail: {\tt fleiner@cs.bme.hu }}}%, Budapest
\and 
Zsuzsanna Jank\'o\thanks{{Corvinus University, Budapest, Hungary. Research was supported by the OTKA K109240
  research project and the MTA-ELTE Egerv\'ary Research Group. E-mail: {\tt
    jzsuzsy@cs.elte.hu}}}
\and
Akihisa  Tamura\thanks{{Keio University, Yokohama, Japan. Research was supported by Grants-in-Aid for
    Scientific Research (B) from JSPS. E-mail: {\tt
      aki-tamura@math.keio.ac.jp}}}\\
\and
Alexander Teytelboym\thanks{{Department of Economics, St. Catherine's College, and the Institute of New Economic Thinking at the Oxford Martin School, University of Oxford, Oxford, United Kingdom. Email: {\tt alexander.teytelboym@economics.ox.ac.uk}}}
}

\date{\today}
\maketitle

%\\\\Preliminary and incomplete draft}

%\ch{The text that is new or has been changed is in blue.}

%
\begin{abstract}
{\small We consider a model of matching in trading networks in which firms can enter into bilateral contracts. In trading networks, \textit{stable} outcomes, which are immune to deviations of arbitrary sets of firms, may not exist. We define a new solution concept called \textit{\trailstabity}. Trail-stable outcomes are immune to consecutive, pairwise deviations between linked firms. We show that any trading network with bilateral contracts has a \trailstab\ outcome
whenever firms' choice functions satisfy the full substitutability condition.  For trail-stable outcomes, we prove results on the lattice structure, the rural hospitals theorem, strategy-proofness, and comparative statics of firm entry and exit. We also introduce \textit{weak} trail stability which is implied by trail stability under full substitutability. We describe relationships between the solution concepts.}
\noindent \begin{flushleft}
\par\end{flushleft}{\small \par}

\noindent \begin{flushleft}
\par\end{flushleft}{\small \par}

\end{abstract}
%\newpage{}

\onehalfspacing
\newpage

\section{Introduction}
Modern production is highly interconnected. Firms typically have a large number of buyers and suppliers and dozens of intermediaries add value to final products before they reach the consumer.
In this paper, we study the structure of contractual relationships between firms. In our model, firms have heterogeneous preferences over sets of bilateral
contracts with other firms. Contracts may encode many dimensions of a
relationship including the price, quantity, quality, and delivery time. The universe of possible relationships between
firms is described by a \textit{trading network}---a multi-sided matching
market in which firms form downstream contracts to sell outputs and
upstream contracts to buy inputs. 

We focus on the existence and structure of stable outcomes in decentralized,
real-world matching markets. In production networks that we consider in
this paper, stable outcomes play the role of equilibrium
and may serve as a reasonable prediction of the outcome of market interactions
\citep{KeCr:82,Roth:84,HaKoNiOsWe:11}.\footnote{The ``market design'' literature has emphasized the importance of the existence of stable
outcomes in order to prevent centralized matching markets from unraveling
\citep{Roth:91}. One market design application of our model is electricity trading in peer-to-peer power systems \citep{morstyn2018bilateral}.}
We obtain a general result: any trading network has an outcome that satisfies a
natural extension of \textit{pairwise stability} \citep{GaSh:62}. Our model of
matching markets subsumes many previous models of matching with contracts,
including many-to-one \citep{GaSh:62,CrKn:81,KeCr:82,HaMi:05} and
many-to-many matching markets
\citep{Roth:84,Soto:99,Soto:04,EcOv:06,KlWa:09}.

We build on a seminal contribution by \citet{Ostr:08}, who introduced a matching model
of \textit{supply chains}. In a supply chain, there are agents, who
only supply inputs (e.g. farmers); agents, who only buy final outputs
(e.g. consumers); while the rest of the agents are intermediaries,
who buy inputs and sell outputs (e.g. supermarkets). All agents are
partially ordered along the supply chain: downstream (upstream)
firms cannot sell to (buy from) firms upstream (downstream) i.e. the
trading network is \textit{acyclic}. His key assumption about the market, which we retain in his paper, was that
firms' choice functions over contracts satisfy  \textit{same-side
  substitutability} and \textit{cross-side complementarity} conditions
(\citet{HaKo:12} later referred to these conditions jointly as \textit{full substitutability}).  This assumption requires that firms view any downstream or any
upstream contracts as substitutes, but any downstream and any upstream
contract as complements.%
\footnote{Same-side substitutability is a fairly strong assumption as, for
  example, it rules out any complementarities in inputs. There is evidence
  that modern manufacturing firms rely on many complementary inputs
  \citep{MiRo:90,Fox:10}. \citet{HaKo:15} and \citet{rostek2017matching} consider a multilateral matching market with complements. \citet{AlTe:15} analyze supply chains in which inputs could be complementary or substitutable with general preferences. \citet{jagadeesan2017complementary} explores trading networks with complementary inputs in a large market setting. While production decisions often create externalities, we assume that firms only care about the contracts they are involved in \citep{SaTo:96, Band:12, Pyci:12, YePy:15}. In addition, we work in a complete information environment; for a treatment with asymmetric information, see \citet{roth1989two, ehlers2007incomplete, chakraborty2010two}. Extending our model to incorporate incomplete information and externalities is a promising area for further research.} \citet{Ostr:08} showed that
any supply chain has a \textit{chain-stable} outcome for which there
are no blocking upstream or downstream chains of contracts. \citet{HaKo:12} further showed that, in the presence
of network acyclicity, chain-stable outcomes are equivalent to
\textit{stable} 
outcomes i.e. those that are immune to deviations by arbitrary sets of firms. Even under full substitutability, stable/chain-stable outcomes in supply chains may be Pareto inefficient.\footnote{Inefficiency arises even in two-sided many-to-many matching markets without contracts if agents have multi-unit demands: \citet{Blai:88} and \citet[p. 177]{RoSo:90} provide the earliest examples; \citet{EcOv:06,KlWa:09} discuss the setting with contracts. \citet{West:10} provides necessary and sufficient conditions on the structure contract relationships in the supply chain for chain-stable outcomes to be efficient.} %These conditions also guarantee that chain-stable outcomes are ``setwise-stable'' \citep{Soto:99,KoUn:06,EcOv:06,KlWa:09}. \citet{West:10} calls ``setwise-stable outcomes'' ``group-stable" and \citet{HaKo:12} call (what we call) ``group-stable" outcomes ``stable".}

While a supply chain may be a good model of production in certain industries
\citep{AnCh:13}, in general, firms simultaneously supply inputs to\textit{
  and} buy outputs from other firms (possibly through intermediaries). If
this is the case, we say a trading network contains a contract \textit{cycle}.  For
example, the sectoral input-output network of the U.S. economy, illustrated
by \citet[Figure 3]{AcCaOzTa:12}, shows that American firms are very
interdependent and the trading network contains many cycles. Consider a
coal mine that supplies coal to a steel factory.  The factory uses coal to
produce steel, which is an input for a manufacturing firm that sells mining
equipment back to the mine. This creates a contract cycle. However, \citet{HaKo:12} 
showed that if a trading network without transfers has a contract cycle then stable
outcomes may fail to exist. Moreover, \citet*{FlJaScTe:17} show that checking whether a stable outcome exists---or even whether a given outcome is stable---is computationally intractable. 

We show that, even in the presence of contract cycles,
outcomes that satisfy a different stability concept---\textit{trail
  stability}---can still be found. A trail of contracts is a set of contracts which can be ordered in such a way that the buyer in one contract is the seller in the subsequent one. Along
a \textit{locally blocking trail}, whenever a firm receives an upstream (downstream) offer, it can either accept it unconditionally or hold the offer and make a myopic, unilateral downstream (upstream) offer. Trail stability rules out any
such locally blocking trails. We argue that trail stability is a useful, natural, and intuitive
equilibrium concept for the analysis of matching markets in networks because locally blocking trails do not require extensive coordination among recontracting parties. In general trading networks, stable and chain-stable outcomes are trail-stable under full substitutability (but not in general). Trail stability is equivalent to chain stability (and therefore to stability under our assumptions) in acyclic trading networks and to pairwise stability in
two-sided matching markets. 

Trail-stable outcomes correspond to the fixed-points of an operator and form a particular lattice structure for \textit{terminal agents} who can sign only upstream or only downstream contracts. The lattice reflects the classic opposition-of-interests property of two-sided markets, but in our case the opposition of interests is between terminal buyers and terminal sellers. In addition to this strong lattice property, we extend previous results on the existence of buyer- and seller-optimal stable outcomes, the rural hospitals theorem, strategy-proofness as well as comparative statics on firm entry and exit that have only been studied in a supply-chain or two-sided setting under general choice functions. 

We then introduce another solution concept called \emph{weak trail stability} which ensures that firms are willing to sign {\it all} their contracts along a \textit{sequentially blocking trail} (rather than simply upstream-downstream pair along a locally blocking trail). This might be a useful solution concept if the blocking contracts are not fulfilled quickly. Stable and chain-stable outcome are always weakly trail-stable. We show the under full substitutability weakly trail-stable outcomes exist because trail-stable outcomes are weakly trail-stable (without full substitutability, however, trail-stable outcomes may not be weakly trail-stable). 

%We then offer sufficient conditions on agents' preferences such that trail-stable and weakly trail-stable or trail-stable and stable outcomes coincide. Trail stable and weakly trail-stable outcomes coincide under an additional assumption of \textit{separability}, a condition that ensures that decisions over certain pairs of upstream and downstream contracts are taken independently from other. Trail-stable and stable outcomes (and hence weakly trail-stable and chain-stable) coincide under an additional assumption of \textit{simplicity}, which stipulates that contracts can be ordered on an ``intensity'' scale (which could represent prices of trades specified in the contracts) and that firms are picking the contracts with the highest intensity.

Our work complements an important recent paper by \citet{HaKoNiOsWe:14} which shows that in general trading networks, under certain conditions, stable outcomes coincide with chain-stable outcomes i.e. those immune to coordinated deviations by a set of firms which is \textit{simultaneously} signing a trail of contracts. Our paper is also related to the stability of (continuous and discrete) network flows discussed by \citet{Flei:09,
  Flei:14}. In these models, agents choose the amount of ``flow'' they receive from upstream
and downstream agents and have preferences over who they receive the
``flow'' from. The network flow model allows for cycles. However, the
choice functions in the network flow models are restricted by Kirchhoff's
(current) law (the total amount of incoming (current) flow is equal to the
total amount of outgoing flow) and in the discrete case these choice
functions are special cases of \citet{Ostr:08}. This paper therefore generalizes both of the supply chain and the network flow models, while offering two appealing new solution concepts. 

\citet{HaKoNiOsWe:11} were the first to consider a general trading network model and proved that stable outcomes always exist under full substitutability in a transferable utility (TU) economy. However, TU rules out wealth effects and distortionary frictions, such as sales taxes, bargaining costs, or incomplete financial markets. \citet*{FlJaJaTe:17} consider a variation of our model in which every contract specifies a trade and a continuous price (see also \citealp{HaKoNiOsWe:14}). They show that, under full substitutability, trail-stable outcomes are essentially equivalent to competitive equilibrium outcomes. In the presence of distortionary frictions, however, they show that competitive equilibria are not stable (in fact, stable outcomes may not exist) and can even be Pareto-comparable.\footnote{Without distortionary frictions, competitive equilibrium outcomes are stable and in the core even in the presence of wealth effects \citep*{FlJaJaTe:17}. In earlier work, \citet{HaKoNiOsWe:11} and \citet{HaKo:15} made this observation in a TU economy.} Hence, trail stability can also serve a cooperative interpretation of competitive equilibrium in settings where price-taking behavior is not a plausible assumption.

We proceed as follows. In Section \ref{modelsec}, we present the ingredients of the
model, including the trading network, main assumptions on choice functions, and terminal agents. In Section \ref{stabconcepts}, we introduce stability, chain stability, and trail stability. Then, in Section \ref{stabsec}, we report the existence of trail-stable outcomes before digging deeper into their structure. In Section \ref{sec:weaktrail}, we introduce weak trail stability and describe its relationship to stability and trail stability. Finally, we conclude and outline some directions for future
work. Appendix \ref{app:previous} summarizes all solution concepts, results and nomenclature used in this paper with reference to previous work and describes which previous results we have generalized in our setting of trading networks with general choice functions. Appendix \ref{app:proofs} provides proofs of the main results. In Appendix \ref{sec:coincide}, we give sufficient conditions on preferences that ensure that trail-stable and stable outcomes or trail-stable and weakly trail-stable outcomes coincide. Appendix \ref{app:path} considers yet another solution concept.

\section{Model}\label{modelsec}

\subsection{Ingredients}

There is finite set of agents (firms or consumers) $F$ and a finite
set of contracts (trading network) $X$.%
\footnote{The standard justification for this assumption is given by \citet[p. 49]{Roth:84}:
``elements of a {[}contract{]} can take on only discrete values;
salary cannot be specified more precisely than to the nearest penny,
hours to the nearest second, etc.'' In fact, the finiteness assumption is not
necessary for our existence proof. We only require that the set of contracts
between any two agents forms a complete lattice.
} A contract $x\in X$ is a bilateral agreement between a buyer $b(x)\in F$
and a seller $s(x)\in F$. A \textit{(trading) cycle} in $X$ is a sequence of firms $(f_1,\ldots,f_M)$ such that for all $m\in\{1,\ldots,M\}$ there exists a contract $x_m$ such that $s(x_m)=f_m$ and $b(x_m)=f_{m+1}$ (subscripts modulo $M$). Hence, $F(x)=\{s(x),b(x)\}$ is the
set of firms associated with contract $x$ and, more generally, $F(Y)$
is the set of firms associated with contract set $Y\subseteq X$.
Denote $X^B_{f}=\{x\in X|b(x)=f\}$ and $X^S_{f}=\{x\in X|s(x)=f\}$
the sets of $f$'s upstream and downstream contracts---i.e. the contract for which
$f$ is a buyer and a seller, respectively. Clearly, $Y^B_{f}$ and
$Y^S_{f}$ form a partition over the set of contracts $Y{}_{f}=\{y\in
Y|f\in F(y)\}$ 
which involve $f$, since an agent cannot be a buyer and a seller
in the same contract.

Our structure is more general than the
setting described by \citet{Ostr:08}, \citet{West:10} or \citet{HaKo:12}. Each firm $f\in F$ is associated
with a vertex of a directed multigraph $(F,X)$ and each contract
$x\in X$ is a directed edge of this graph. For any $f$, the set $X^B_{f}$
corresponds to the set of incoming edges and $X^S_{f}$ to the set of outgoing edges. An \textit{acyclic} trading network (a \textit{supply chain}) contains no directed cycles in the graph. Supply chains require a partial order on
the firms' positions in the chain although firms may sell to (buy
from) any downstream (upstream) level. In our model, we consider general trading networks,
which may contain contract cycles (i.e. directed cycles on the graph).

%\begin{figure}
%\begin{minipage}{.5\textwidth}
%\begin{centering}
%\par\end{centering}
%\centering{}
%\scalebox{1}{\includegraphics{supplychain}}
%\protect\caption{Supply chain\label{fig:Supply-chain}}
%\end{minipage}
%%\end{figure}
%%\begin{figure}
%\begin{minipage}{.5\textwidth}
%\begin{centering}
%\par\end{centering}
%\centering{}
%\scalebox{1}{\includegraphics{network}}
%\protect\caption{Trading network\label{fig:Contract-network}}
%\end{minipage}
%\end{figure}

Every firm has a choice function $C^{f}$, such that $C^{f}(Y_{f})\subseteq Y_{f}$
for any $Y_{f}\subseteq X_{f}$.\footnote{Since firms only care about their own contracts, we can write $C^{f}(Y)$ to mean $C^{f}(Y_{f})$.} We say that choice functions of $f\in F$ satisfy the \textit{irrelevance of rejected contracts (IRC)} condition if for any $Y\subseteq X$ and $C^{f}(Y)\subseteq Z\subseteq Y$ we have that $C^{f}(Z)=C^{f}(Y)$ \citep{Blai:88,Alka:02,Flei:03,Eche:07,AySo:12}.\footnote{In our setting, IRC is equivalent to the Weak Axiom of Revealed Preference \citep{AySo:12,Alva:15a}.}

For any $Y\subseteq X$ and $Z\subseteq X$, define the \textit{chosen}
set of upstream contracts 
\begin{equation}
C_{B}^{f}(Y|Z)= C^{f}(Y^B_f\cup Z^S_f)\cap X^B_f
\end{equation}
which is the set of contracts $f$ chooses as a buyer when $f$ has
access to upstream contracts $Y$ and downstream contracts $Z.$ Analogously,
define the chosen set of downstream contracts 
\begin{equation}
C_{S}^{f}(Z|Y)= C^{f}(Z^S_f \cup Y^B_f)\cap X^S_f
\end{equation}
Hence, we can define \textit{rejected} sets of contracts $R_{B}^{f}(Y|Z)= Y_f\setminus C_{B}^{f}(Y|Z)$
and $R_{S}^{f}(Z|Y)= Z_f\setminus C_{S}^{f}(Z|Y)$. An \textit{outcome}
$A\subseteq X$ is a set of contracts.

A set of contracts $A\subseteq X$ is \textit{individually rational} for an agent $f\in F$ if $C^{f}(A_{f})=A_{f}$. We call set $A$ \textit{acceptable} if $A$ is 
individually rational for all agents $f\in F$. For sets of
contracts $W,A\subseteq X$, we say that $A$ is \textit{$(W,f)$-acceptable} if
$A_{f}\subseteq C^f(W_{f}\cup A_{f})$ i.e. if the agent $f$ chooses all
contracts from set $A_f$ whenever she is offered $A$ alongside
$W$. Set of contracts $A$ is \textit{$W$-acceptable} if $A$ is $(W,f)$-acceptable
  for all agents $f\in F$. Note that contract set $A$ is individually rational for agent $f$ if and only if it is
  $(\emptyset,f)$-acceptable. If $y\in X^B_f$ and $z\in X^S_f$ then $\{y,z\}$ is a
  \textit{$(W,f)$-essential pair} if neither $\{y\}$ nor $\{z\}$ is $(W,f)$-acceptable
  but $\{y,z\}$ is $(W,f)$-acceptable. Note that any essential pair consists
  of exactly one upstream and one downstream contract.

\subsection{Assumptions on choice functions}

We can now state our key assumption on choice functions introduced by \citet{Ostr:08} and \citet{HaKo:12}.
\begin{definition}

Choice functions of $f\in F$ are \textit{fully
substitutable} if for all $Y'\subseteq Y\subseteq X$ and $Z'\subseteq Z\subseteq X$
they are:
\begin{enumerate}
\begin{minipage}{.5\textwidth}
\item \textit{Same-side substitutable} (SSS):

\begin{enumerate}
\item $R_{B}^{f}(Y'|Z)\subseteq R_{B}^{f}(Y|Z)$
\item $R_{S}^{f}(Z'|Y)\subseteq R_{S}^{f}(Z|Y)$
\end{enumerate}
\end{minipage}
\begin{minipage}{.5\textwidth}
\item \textit{Cross-side complementary} (CSC):

\begin{enumerate}
\item $R_{B}^{f}(Y|Z)\subseteq R_{B}^{f}(Y|Z')$
\item $R_{S}^{f}(Z|Y)\subseteq R_{S}^{f}(Z|Y')$
\end{enumerate}
\end{minipage}
\end{enumerate}
\end{definition}

Contracts are fully substitutable if every firm regards any of its
upstream or any of its downstream contracts as substitutes, but its
upstream and downstream contracts as complements. Hence, rejected
downstream (upstream) contracts continue to be rejected whenever the
set of offered downstream (upstream) contracts expands or whenever
the set of offered upstream (downstream) contracts shrinks.

\subsection{Laws of Aggregate Demand and Supply}
We first re-state the familiar Laws of Aggregate Demand and Supply (LAD/LAS) \citep{HaMi:05,HaKo:12}. LAD (LAS) states that when a firm has more upstream (downstream) contracts available (holding the same downstream (upstream) contracts), the number of downstream (upstream) contracts the firms chooses does not increase more than the number of upstream (downstream) contracts the firm chooses. Intuitively, an increase in the availability of contracts on one side, does not increase the difference between the number of contracts signed on either side.\footnote{This is an extension of the ``cardinal monotonicity'' condition \citep{Alka:02}.}
\begin{definition}

Choice functions of $f\in F$ satisfy the \emph{Law of Aggregate Demand} if for all $Y,Z\subseteq X$
and $Y'\subseteq Y$

\[
|C_{B}^{f}(Y|Z)|-|C_{B}^{f}(Y'|Z)|\geq|C_{S}^{f}(Z|Y)|-|C_{S}^{f}(Z|Y')|
\]

and the \emph{Law of Aggregate Supply} if for all $Y,Z\subseteq X$ and $Z'\subseteq Z$

\[
|C_{S}^{f}(Z|Y)|-|C_{S}^{f}(Z'|Y)|\geq|C_{B}^{f}(Y|Z)|-|C_{B}^{f}(Y|Z')|
\]

\end{definition}

We can easily show that LAD/LAS imply IRC, extending the result by \citet{AySo:12}.

\begin{lemma}\label{ladlasirc} Suppose that choice function of $f\in F$ satisfies full substitutability and LAD/LAS. Then the choice function of $f$ satisfies IRC.
\end{lemma}
\begin{proof}[Proof of Lemma \ref{ladlasirc}]
Consider $Y\subset X_f$ and $z\in X_f^B \setminus Y$ such that $z\notin C^f(Y\cup\{z\})$. Then (1), from SSS, we have that $C_B^f(Y\cup\{z\}) \subseteq C_B^f(Y) $, and (2), from CSC, we have that  $C_S^f(Y\cup\{z\}) \supseteq C_S^f(Y)$. If choice functions satisfy LAD/LAS, then we have that
$|C_{B}^{f}(Y)|-|C_{S}^{f}(Y))|\leq|C_{B}^{f}(Y\cup\{z\})|-|C_{S}^{f}(Y\cup\{z\})|$.
But then both of the above the set inclusions (1) and (2) must hold with equality, so $C^f (Y\cup\{z\}) =C^f(Y)$ as required.
\end{proof}

\subsection{Terminal agents and terminal superiority}
We now introduce some terminology that describes contracts of agents, who only act as buyers or only act as sellers. A firm $f$ is a \textit{terminal seller} if there are no upstream contracts for $f$ in
the network and $f$ is a \textit{terminal buyer} if the network does not contain
any downstream contracts for $f$. An agent who is either a terminal buyer or a terminal seller is called a \textit{terminal agent}. Let $\mathcal{T}$ denote the set of terminal agents in $F$ and for a set $A$ of contracts let 
us denote the \emph{terminal contracts of $A$} 
\ch{ by $A_\mathcal{T}=\bigcup \{A_f|f\in \mathcal{T}\}$. }A set $Y$
of contracts is \emph{terminal-acceptable} if there is an acceptable set
$A$ of contracts such that $Y=A_\mathcal{T}$. If $A$ and $W$ are terminal-acceptable sets of contracts
then we say that $A$ is \textit{seller-superior} to $W$ (denoted by
$A\succeq^S W$) if $C_f(A_f\cup W_f)=A_f$ for each terminal seller $f$ and
$C_g(A_g\cup W_g)=W_g$ for each terminal buyer $g$. Similarly, $A$ is
\textit{buyer-superior} to $W$ (denoted by $A\succeq^B W$) if $C_f(A_f\cup
W_f)=W_f$ for each terminal seller $f$ and $C_g(A_g\cup W_g)=A_g$ for each
terminal buyer $g$.\footnote{The superiority partial order is equivalent to the revealed preference relation \citep{alkan2003stable,chambers2017choice}.} Clearly, these relations are opposite, that is, $W\succeq^S
A$ if and only if $A\succeq^B W$ holds. Whenever either relation holds, we call this partial order on outcomes \textit{terminal superiority}. \ch{Terminal agents are going to play a key role when we describe the structure of outcomes in trading networks.}

\subsection{Solution concepts}\label{stabconcepts}

We start off by defining two solution concepts that have appeared in previous work.

\begin{definition}[\citealt{HaKo:12}] An outcome $A\subseteq X$ is \textit{stable}\footnote{\citet{KlWa:09} consider ``weak setwise stable'' outcomes which are immune to blocks in which sets of agents must also agree on which contracts they can drop. \citet{West:10} considers ``group-stable'' outcomes which are immune to blocks in which sets of agents can sign better (rather the best) contracts (also known as ``setwise stable outcomes'', see \citealp{Soto:99,EcOv:06,KlWa:09}).} 
\begin{enumerate}
\item $A$ is acceptable.
\item There exist no non-empty blocking set of contracts $Z\subseteq X$, such that $Z\cap A=\emptyset$ and $Z$ is $(A,f)$-acceptable for all $f\in F(Z)$.
\end{enumerate}
\end{definition}

Stable outcomes are immune to deviations by sets of firms, which can re-contract freely among themselves while keeping any of their existing contracts. Stable outcomes always exist in acyclic networks if choice functions are fully substitutable. In order to study more general trading networks, we first introduce trails of contracts.

\begin{definition}
A non-empty set of contracts $T$ is a
\textit{trail} if its elements can be arranged in some order $(x_{1},\ldots,x_{M})$ such that
$b(x_{m})=s(x_{m+1})$ holds for all $m\in \{1,\ldots,M-1\}$ where $M=|T|$.
\end{definition}

While a trail may not contain the \ch{same contract more than
once}, it may include the same agents any number of times.
For example, in Figure
\ref{fig:hatkom}, there is a trail $\{w,z,y\}$ that starts from firm $m$ (or $j$), ends at firm $j$ (or $m$), and ``visits'' firm $k$. 

\begin{definition}[\citealt{HaKoNiOsWe:14}] An outcome $A\subseteq X$ is \textit{chain-stable} if 
\begin{enumerate}
\item $A$ is acceptable.
\item There is no trail $T$, such that $T\cap A=\emptyset$ and $T$ is $(A,f)$-acceptable for all $f\in F(T)$.
\end{enumerate}
\end{definition}

\citet{HaKoNiOsWe:14} show that in general trading networks stable
outcomes are equivalent to chain-stable outcomes whenever choice functions satisfy full substitutability and Laws of
Aggregate Demand and Supply.\footnote{\citet{HaKoNiOsWe:14} refer to trails
 as ``chains'' hence the name ``chain stability''. This extends the definition of chain stability introduced by \citet{Ostr:08}.} However, \citet{Flei:09} and
\citet{HaKo:12} show that stable outcomes may not exist in
general trading networks (see also Example \ref{trail-not-group-stable-example}
below). Moreover, \citet*{FlJaScTe:17} show that both problems of determining whether a stable outcome exists and determining whether an outcome is stable are computationally intractable. 

%Let us define decision problem GS as
%follows. An instance of GS is a 
%trading network with a set of agents $F$ and set of contracts $X$ (with choice functions that satisfy full substitutability and IRC) and an outcome $A$. The answer for an instance of GS is YES if the particular outcome $A$ is not stable (that is, if there is a set of contracts $Z$ that blocks $A$),
%otherwise the answer is NO. 
%%DEFINE ORACLE; DEFINE NP-COMPLETE.  
%\begin{theorem}\label{grouphopeless}
%Problem $GS$ is NP-complete. Moreover, if choice functions are represented by
%oracles then finding the right answer for an instance of GS might need an
%exponential number of oracle calls.
%\end{theorem}

The non-existence and computational intractability of stable outcomes motivates us to define an alternative solution concept.
We first define \textit{trail stability}, which coincides with pairwise stability
in a two-sided many-to-many matching market with contracts and with chain stability in supply chains.

\begin{definition} An outcome $A\subseteq X$ is \textit{trail-stable} if 
\begin{enumerate}
\item $A$ is acceptable.
\item There is no trail $T=\{x_1,x_2,\ldots,x_M\}$, such that $T\cap
  A=\emptyset$  and
\begin{enumerate}
\item$\{x_{1}\}$ is $(A,f_1)$-acceptable for $f_{1}=s(x_{1})$, and
\item$\{x_{m-1}, x_{m}\}$ is $(A,f_m)$-acceptable for
$f_{m}=b(x_{m-1})=s(x_{m})$ whenever $1<m\le M$ and
\item$\{x_{M}\}$ is $(A,f_{M+1})$-acceptable for $f_{M+1}=b(x_{M})$.
\end{enumerate}
The above trail $T$ is called a \textit{locally blocking trail to
  $A$}.
\end{enumerate}

\end{definition}

Trail stability is a natural solution concept when firms interact
mainly with their buyers and suppliers and deviations by arbitrary
sets of firms are difficult to arrange. In a trail-stable outcome,
no agent wants to drop his contracts and there exists no sequence
 of \textit{consecutive}
bilateral contracts comprising a trail preferred by all the agents
in the trail to the current outcome. 
First, $f_{1}$ makes an unilateral offer of $x_{1}$ (the first contract
in the trail) to the buyer $f_{2}$.
The buyer $f_{2}$ then either unconditionally accepts the offer (forming
a locally blocking trail) or conditionally accepts the seller's offer while
looking to make a contract offer ($x_{2}$) to another buyer $f_{3}$.
If $f_{2}$'s buyer in $x_{2}$ happens to be $f_{1}$, then $f_{1}$
considers the offer of $x_{2}$ on it own. If $f_{1}$ accepts, we have a locally blocking trail. If $f_{2}$'s
buyer is not $f_{1}$, then his buyer $f_{3}$ either accepts $x_{2}$ unconditionally
or looks for another seller $f_{4}$ after a conditional acceptance
of $x_{2}$. The trail of these linked conditional contract offers continues until
the last buyer $f_{M+1}$ in the trail unconditionally accepts the
upstream contract offer $x_{M}$.%
\footnote{The trail and the order of conditional acceptances can, of course,
be reversed with $f_{M+1}$ offering the first upstream contract to
seller $f_{M}$ and so on. %
}. Note that all intermediate agents only need to myopically decide whether they want to choose pairs of upstream-downstream contract that ``pass though'' the agents along the trail. In other words, whenever an agent receives a contract offer, he is ``activated'' to either accept or reject it or to hold this offer while making his own contract offer. 

The following example illustrates that trail-stable outcomes are not necessarily stable.\footnote{Example \ref{trail-not-group-stable-example} is similar to examples in \citet[p. 12]{Flei:09} and \citet[Fig. 3, p. 13]{HaKo:12}.}

\begin{figure}
\begin{minipage}{.60\textwidth}

\centering\scalebox{1.5}{\xymatrix{
& j \ar@/_/[d]_z \ar@/^1.25pc/[rd]^{x}\\
m \ar@/^1.25pc/[ru]^{w} & k \ar@/_/[u]_y & i
}}
\protect\caption{Trading network in Examples 1, 3, and 4.\label{fig:hatkom}}
\end{minipage}
\begin{minipage}{.45\textwidth}
\centering\scalebox{1.5}{\xymatrix{
 j \ar@/_/[d]_z \\
 k \ar@/_/[u]_y
}}
\protect\caption{Trading network in Example 2.\label{fig:cycle}}
\end{minipage}

\end{figure}

%%\ch{
\begin{example}[Trail-stable outcomes are not necessarily stable] \label{trail-not-group-stable-example}
Consider four contracts $x$, $y$, $z$ and $w$. Assume that $i=b(x)$, $j=s(x)=s(z)=b(y)=b(w)$,
$k=b(z)=s(y)$ and $m=s(w)$ (see Figure \ref{fig:hatkom}).
Agents have the following preferences that induce fully substitutable choice functions:\footnote{In all our examples, $\succ$ denotes a strict preference relation. Choice functions induced by strict preferences satisfy IRC. \info{We say that $f \in F$ ``prefers'' outcome $A$ to outcome $A'$ if $C^f(A\cup A')=A_f$} }\\
\begin{align*}
\succ_{i}&:\{x\}\succ_{i}\emptyset\\
\succ_{m}&:\{w\}\succ_{m}\emptyset\\
\succ_{j}&:\{x,y,w\}\succ_{j}\{z,y,w\}\succ_{j}  \{x,y\}\succ_{j}\{z,y\}\succ_{j}\{w\}\succ_{j}\emptyset\\
\succ_{k}&:\{z,y\}\succ_{k}\emptyset
\end{align*}
and other outcomes are not acceptable.
A trail-stable outcome exists: $A=\{w\}$. The trail-stable outcome $\{w\}$ is Pareto-inefficient as $\{z,y,w\}$ makes $j$ and $k$ better off without making $i$ and $m$ worse off. There is, however, no stable outcome.%
\footnote{One can check that there is no stable outcome since $\{w\}\succ_{j} \{x,w\}\succ_{k} \{x,z,w\}\succ_{i,j} \{z,y,w\}\succ_{j,k} \{w\}$
and other outcomes are not acceptable.%
} 
\end{example}
%%}

To illustrate trail stability further, let us drop agents $i$ and $m$ and their corresponding contracts from
the example above (Figure \ref{fig:cycle}). % while keeping preferences of the remaining agents the same.
\ch{ The  new preferences of $j$ are $\{y,z\}\succ_{j}\emptyset$. }  There is one stable outcome $\{y,z\}$. There are, however,
two weakly trail-stable outcomes: $\emptyset$ and $\{y,z\}$. Is $\emptyset$
a reasonable possible outcome of this market? If coordination is difficult, due to bargaining costs for example, firms may find it difficult to pin down blocking sets, especially in large markets. Trail stability therefore provides a natural solution concept for matching markets in which firms have a limited ability to coordinate their decisions in the trading network.

%\ch{
\section{Existence and properties of trail-stable outcomes}\label{stabsec}
We can now state the first key result of this paper.
\begin{theorem}\label{main}
Suppose that choice functions satisfy full substitutability and IRC. Then there exists a trail-stable outcome.\footnote{Our results do not contradict Theorem 5 on the non-existence of stable outcomes in \citet{HaKo:12} since Theorem \ref{main} only considers the existence of trail-stable outcomes.} 
\end{theorem}
This theorem establishes a positive existence result for an appealing solution concept in general trading networks: under the usual assumption of full substitutability, trail-stable outcomes always exist.%In order to examine the structure of weakly trail-stable outcomes, recall that in the marriage model of Gale and Shapley, the existence of man-optimal and woman-optimal stable matchings follow from the well-known lattice structure of stable matchings. The key to extending this result to trading networks is to consider only terminal agents. In a supply chain setting, \citet{HaKo:12} show that there exists a stable outcome $A^{\star}$ that is \textit{buyer-optimal} (\textit{seller-optimal}) i.e. one in which if all terminal buyers (terminal sellers) unanimously prefer $A^{\star}$ to any other stable outcome.\footnote{For any terminal-buyer (terminal-seller) $f\in F$ and any for any stable $Z\subseteq X$, we have that $C^f(A^{\star}_f\cup Z)=A^{\star}_f$. This is a common property of stable outcomes in two-sided markets with substitutable choice functions, however, it typically fails in richer matching models \citep{YePy:15,Alva:15, AlTe:15}.} To explore the structure of weakly trail-stable outcomes, we will need to invoke \ch{separability}. We say that $Y\subseteq X$ is \emph{terminal-weakly trail-stable} if there is a weakly trail-stable outcome $A\subseteq X$ such that $Y=A_\mathcal{T}$.
%\begin{theorem}\label{lattice}
%In any trading network $X$ if choice functions of $F$ satisfy full substitutability, IRC and separability then   \ch { the set of weakly trail-stable outcomes contains buyer-optimal and seller-optimal outcomes.} 
%\info{omit: the terminal-weakly trail-stable contract sets form a lattice under buyer-superiority.}
%\end{theorem}

%\ch {Theorem \ref{lattice} extends Theorem 4 by \citet{HaKo:12}, which establishes the existence of buyer- and seller-optimal outcomes in acyclic trading networks.}
%}
%%%%
%\iffalse
%%%%

%\subsection{Structure of trail stable outcomes in general
%  networks}\label{ratsec} 

In order to prove Theorem \ref{main}, we use tools familiar to matching theory, such as the Tarski fixed-point theorem \citep{Adac:00,Flei:03,EcOv:06,HaMi:05, Ostr:08, HaKo:12}.
\ch{Let $X^{B}$ and $X^{S}$ be two subsets of $X$ which represent the set of contracts offered to buyers and sellers.  We define an isotone operator $\Phi$ that acts on $(X^{B},X^{S})$ and show that any fixed-point $(\dot X^{B},\dot X^{S})$ of  $\Phi$  corresponds to a {trail}-stable outcome $A=\dot X^{B}\cap \dot X^{S}$.} \ch{These tools allow us to explore properties of trail-stable outcomes that have previously only been explored in a supply-chain or a two-sided setting.}

As we have already seen, trail-stable outcomes are not always stable or chain-stable even under full substitutability. However, the converse is true.

\begin{proposition}\label{prop:stable-is-trail-stable}
Suppose that choice functions satisfy full substitutability and IRC. Then any stable/chain-stable outcome is trail-stable.
\end{proposition}
Full substitutability is key to this result as the following example shows.

\begin{example}[Stability and chain stability do not imply trail stability without full substitutability] \label{ex:not-fully-weakly trail-stable}
Consider two contracts $y,z$. Assume that $j=s(z)=b(y)$,
$k=b(z)=s(y)$ (see Figure \ref{fig:cycle}).
Agents have the following preferences:\\
\begin{align*}
\succ_{j}&:\{z\}\succ_{j}\{y\}\succ_{j}\emptyset\\
\succ_{k}&:\{z,y\}\succ_{k}\emptyset.
\end{align*}
and other outcomes are not acceptable. Note that $k$'s preference are fully substitutable, but $j$'s preferences are not. The empty set of contracts is stable (and chain-stable), however, it is not trail-stable since $\{z,y\}$ is a locally blocking trail.
\end{example}

This example highlights that locally blocking trails need not be acceptable (alongside other contracts) themselves: firm $j$ ``forgets'' that it offered contract $z$ when it considers the terminal contract $y$ offered in return by $k$. In Section \ref{sec:weaktrail}, we will consider a solution concept which ensures that agents are willing to accept all the contracts along a locally blocking trail (perhaps alongside other contracts).

\subsection{Structure of trail-stable outcomes}\label{ratsec} 
Recall that in the marriage model of Gale and Shapley, the existence of man-optimal and woman-optimal stable matchings follow from the well-known lattice structure of stable matchings. The key to extending this result to trading networks is to consider only terminal agents. We say that a trail-stable outcome $A_{max}$ ($A_{min}$) that is \textit{buyer-optimal} (\textit{seller-optimal}) if any terminal buyer (terminal seller) prefers it to any other outcome i.e. for any trail-stable $Z\subseteq X$, we have that $C^f(A_{max}\cup Z)=A_{f,max}$.
\begin{proposition}\label{weaker-lattice}
Suppose that choice functions satisfy full substitutability and IRC. Then the set of trail-stable outcomes contains buyer-optimal and seller-optimal outcomes.
\end{proposition}
Proposition \ref{weaker-lattice} extends Theorem 2 by \citet{Ostr:08} and Theorem 4 by \citet{HaKo:12}, which establish the existence of buyer- and seller-optimal outcomes in acyclic trading networks.\footnote{This is a common property of stable outcomes in two-sided markets with substitutable choice functions, however, it typically fails in richer matching models \citep{YePy:15,Alva:15, AlTe:15}.}
%\ch{Lemma \ref{weaker-lattice}} is simply an analogue of %Theorem \ref{lattice} for trail-stable outcomes. 
We say that $Y\subseteq X$ is \emph{terminal-trail-stable} if there is a trail-stable outcome $A\subseteq X$ such that $Y=A_\mathcal{T}$.
\begin{proposition}\label{fully-stable-lattice}
Suppose that choice functions satisfy full substitutability and LAD/LAS.  Then the terminal-trail-stable contract sets form a lattice under terminal superiority. \end{proposition}

Proposition \ref{fully-stable-lattice} shows that whenever LAD/LAS holds choice functions of terminal agents define a natural partial order on outcomes and the terminal-trail-stable contract sets form a lattice under this order.\footnote{In fact, our proof shows that the terminal-trail-stable contract sets form a \emph{sublattice} \citep{Blai:88,Flei:03,EcOv:06}.} Note that for the lattice and the opposition-of-interests structure, only terminal agents play a role: two outcomes are equivalent if all the terminal agents have the same set of contracts. Indeed, if $A^1$ and $A^2$ are trail-stable
outcomes then there is a trail-stable outcome $A^{+}$ such that all
terminal buyers prefer $A^{+}$ to both $A^1$ and $A^2$ and all sellers prefer any of
$A^1$ and $A^2$ to $A^{+}$.\footnote{Of course, the same holds for if we exchange the
role of buyers and sellers.} This establishes full ``polarization of interests'' in trail-stable outcomes in the sense of \citep{Roth:85} and immediately implies the existence of buyer-optimal ($A_{max}$) and seller-optimal ($A_{min}$) trail-stable outcomes. Therefore, our result substantially strengthens and generalizes the previous results by \citet{Roth:85,Blai:88,EcOv:06} and \citet{HaKo:12}.

The lattice structure of fully-trail stable outcomes allows us to straightforwardly extend two well-known properties of stable outcomes that have been known in two-sided matching markets and acyclic trading networks. One such property is the classic ``rural hospitals theorem'', which shows that in every stable allocation of a two-sided many-to-one doctor-hospital matching market, the same number of doctors are matched to every hospital \citep{Roth:86}. In buyer-seller networks, we can instead consider the difference between the number of upstream and downstream contracts that firms sign \citep{HaKo:12}. The following proposition gives the most general rural hospital theorem result.

\begin{proposition}\label{rural}
Suppose that choice functions satisfy full substitutability and LAD/LAS. Then, for each firm, the difference between the number of upstream contracts and the number of downstream contracts is invariant across trail-stable outcomes.\footnote{Theorem 4 in \citet{Flei:14}, which states that any two stable flows agree on terminal contracts and any two stable flows have the same number of assignments, is a further strengthening of Propositions \ref{fully-stable-lattice} and Propositions \ref{rural} in the special case of network flows.}
\end{proposition}

The lattice structure of trail-stable outcomes also gives a (somewhat weak) mechanism design result.\footnote{One design application of trading networks is a peer-to-peer electricity market in which many consumers also generate electricity \citep{morstyn2018bilateral}.} A mechanism $\mathcal{M}$ is a mapping from a profile of agents' choice functions, $\mathbf{C}^F=(C^f)_{f\in F}$, to the set of outcomes. 
\begin{definition}
A mechanism is \textit{group strategy-proof} for $G\subseteq F$ if for any $\bar{G}\subseteq G$, there does not exist a choice function profile $\mathbf{\bar{C}}^{\bar{G}}$ such that for outcomes $\bar{A}=\mathcal{M}(\mathbf{\bar{C}}^{\bar{G}}, \mathbf{C}^{F\setminus \bar{G}})$ and $A=\mathcal{M}(\mathbf{C}^{F})$ we have that $C^f(\bar{A}\cup A)=\bar{A}$ for \ch{every} $f\in \bar{G}$.
\end{definition}
A mechanism is group strategy-proof for a group of agents if they cannot jointly manipulate their choice functions and obtain an outcome that is better for all of them. Like \citet{HaKo:12}, we are only going to consider group strategy-proofness for terminal agents. We generalize their Theorem 10 with the following result.
\begin{proposition}\label{Strategyproof}Suppose that choice functions satisfy full substitutability and LAD/LAS and, additionally, all terminal buyers (terminal sellers) demand at
most one contract. Then any mechanism that selects the buyer-optimal (seller-optimal)
trail-stable outcome is group strategy-proof for all terminal buyers. \end{proposition}
As is well known, the assumptions that underpin Proposition \ref{Strategyproof}---unit demands and extreme one-sidedness---cannot be substantially relaxed \citep{HaKo:09}.

\subsection{Trail-stable outcomes and comparative statics}\label{}

The second set of properties of trail-stable outcomes concerns the effect of entry and exit of new firms in the trading network. This type of comparative static analysis is well-studied in two-sided matching markets \citep{GaSo:85a,Craw:91,BlRoRo:97,HaMi:05}. More recently, \citet{Ostr:08} and \citet{HaKo:13} extended these results the case of supply chains.
 
%From Lemma \ref{weaker-lattice}, we know  that "in any trading network X if choice functions of F satisfy full substitutability and IRC then the set of trail-stable outcomes contains buyer-optimal and seller-optimal outcomes." As we have seen in the proof of Lemma \ref{main} and Lemma \ref{weaker-lattice}, the fixed points of $\Phi$ form a lattice, hence there is a $\sqsubseteq$-minimal fixed point $(Y^B,Y^S)$ and a $\sqsubseteq$-maximal one $(Z^B,Z^S)$. We showed that trail-stable outcome $A^Y$ is seller-optimal and $A^Z$ is buyer optimal. In the following let us denote  $A^Y$ as  $A_{min}$ and  $A^Z$  as  $A_{max}$  

First,  let us consider what happens when a terminal seller is added to the
market. %Of course, removing a supplier of basic inputs or a consumer of final outputs has the opposite effect on the remaining nodes.
More formally, let $F'= F \cup \{f'\}$  and let $A'_{max}$   and $A'_{min}$  be the buyer-optimal and the seller-optimal trail-stable outcomes in $F'$ respectively.

\begin{proposition}\label{ostrovsky}

Suppose that choice functions satisfy full substitutability and IRC.
Suppose moreover that a new terminal seller $f'$ whose choice function is fully substitutable and satisfies IRC enters the market. Then 
\begin{itemize}
\item every terminal seller $f\neq f'$ prefers $A_{max}$ to $A'_{max}$ and prefers $A_{min}$ to $A'_{min}$, and
\item every terminal buyer $f$ prefers $A'_{max}$ to $A_{max}$ and prefers $A'_{min}$ to $A_{min}$.\footnote{The opposite holds when $f'$ is a terminal buyer.}
\end{itemize}

%\ch{Moreover, we can state that for the intermediate agents, %everyone's buying position gets better, and their selling position gets worse.}
\end{proposition}

Proposition \ref{ostrovsky} says that with a new seller, the seller-optimal outcome $A_{min}$ and the buyer-optimal outcome $A_{max}$ move in the direction favorable to terminal buyers and unfavorable to terminal sellers. Symmetrically, when a terminal buyer is
added or if a seller leaves, $A_{min}$ and $A_{max}$ move in the opposite direction. In other words, more
competition on one end of an industry is bad for the agents on that end and good for the agents
on the other end.\footnote{It is also possible to prove an analogous result to Proposition \ref{ostrovsky} when choice function of terminal agents expands under IRC and full substitutability. Choice function of a terminal agent $\widehat{C}^f$ on $2^X$ is an expansion of choice function $C^f$ on $2^X$ if, for every $Y\subseteq X$, $\widehat{C}^f(Y)\subseteq C^f(Y)$ \citep{echenique2015control,chambers2017choice}. Then suppose all choice functions satisfy full substitutability and IRC and the choice function of one of the terminal buyers (sellers) expands. Then every terminal seller (buyer) prefers the new buyer-optimal and seller-optimal trail-stable outcomes to the old ones and vice versa for the buyers (sellers). This is a straightforward adaptation of Theorem 2 in \citep{chambers2017choice}.} Proposition \ref{ostrovsky} generalizes Theorem 3 in \citet{Ostr:08}. 

Now consider the following \textit{market readjustment process}: %we mean the following:
 When the new terminal seller $f'$ enters, \ch{ and we already have a trail-stable outcome $A$ with corresponding fixed point $(\dot X^B, \dot X^S)$  then let $X$ be the set of all contracts in the new network, and let us define $(\dot X'^{B},\dot X'^{S})=(\dot X^B, \dot X^S\cup X_{f'})$. Operator $\Phi' $ acts on $(X'^{B},X'^{S})$ using choice functions of $F'$.
Let $(\hat X^B, \hat X^S)$ be the fixed point of the iteration of fuction $\Phi$,  
with associated outcome
$ \hat A=\hat X^B \cap \hat X^S$ which is the result of the market readjustment process. 
} 

\begin{proposition}\label{vacancy}
Suppose that choice functions satisfy full substitutability and IRC.
Consider a trail-stable outcome $A$ with associated buyer and seller offer sets $X^B$ and $X^S$.
Suppose a new terminal seller $f'$ whose choice function is fully substitutable and satisfies IRC enters the market and let $ \hat A$ be the result of the market readjustment process.  Then, all terminal sellers prefer $A$ to $ \hat A$ and all terminal buyers  (other than $f'$) prefer $\hat A$ to $A$.\footnote{The opposite holds when $f'$ is a terminal buyer.
}
\end{proposition}
An analogous result is obtained when terminal buyers and terminal sellers exit the market so Proposition \ref{vacancy} generalizes the Theorem in \citet{HaKo:13}.

\section{Weak trail stability}\label{sec:weaktrail}
As we saw in the previous section, the entire locally blocking trail does not need to be acceptable for agents who are participating in the block. If contracts are not fulfilled immediately, then it might be important to ensure that agents want to select all the contracts along a blocking trail. Let us consider a trail $T=\{x_{1},...,x_{M}\}$ whose elements are arranged in a sequence $(x_{1},...,x_{M})$ and define $T_{f}^{\le m}=\{x_{1},...,x_{m}\}\cap T_{f}$ to be firm $f$'s contracts out of the first $m$ contracts in the trail
and $T_{f}^{\ge m}=\{x_{m},...,x_{M}\}\cap T_{f}$ to be firm $f$'s contracts out of the last $M-m+1$
contracts in the trail (where $m \in \{1,\ldots,M\}$).

%Suppose contracts only need to be fulfilled sequentially i.e. once a firm's upstream contract has been fulfilled, it immediately fulfils its downstream contract. (Alternatively, contracts further down the trail could be specified to be fulfilled later.) This is a natural assumption in sequential production networks as production may not be able to continue without inputs and inputs would not be bought without a standing order. Then firms do not need to worry about being involved in multiple chains of contracts along the trail since they never need to be fulfilled together. As such trail stability can be a useful solution concept in production networks in which production is sequential rather than (possibly) simultaneous. Trail stability may be a better solution concept for a short-run prediction of network stability whereas weak trail stability is more suitable for the long run.  It turns out that trail-stable outcomes also exist in general production networks.
\begin{definition}\label{trailstable} An outcome $A\subseteq X$ is \textit{weakly trail-stable} if 
\begin{enumerate}
\item $A$ is acceptable.
\item There is no trail $T=\{x_1,x_2,\ldots,x_M\}$, such that $T\cap
  A=\emptyset$  and %%either (a),(b1),(c) or (a),(b2),(c) holds:
\begin{enumerate}
\item $\{x_1\}$ is $(A,f_1)$-acceptable for $f_1=s(x_{1})$ and
\item At least one of the following two options holds:
\begin{enumerate}
\item   $T_{f_{m}} ^{\le m}$ is 
$(A,f_m)$-acceptable for $f_m=b(x_{m-1})=s(x_{m})$ whenever $1<m\le M$, or
\item $ T_{f_{m}}^{\ge m-1}$ is
$(A,f_m)$-acceptable for $f_m=b(x_{m-1})=s(x_{m})$ whenever $1<m\le M$
\end{enumerate}
\item $\{x_M\}$ is $(A,f_{M+1})$-acceptable for $f_{M+1}=b(x_{M})$.
\end{enumerate}
The above trail $T$ is called a \textit{sequentially blocking trail to $A$}.

\end{enumerate}
\end{definition}

The agents who are participating in a sequentially blocking trail need to choose all the contracts in the trail whenever the trail ``loops back'' to them. As the sequentially blocking trail grows, we ensure that each intermediate
agent wants to choose \textit{all} his contracts along the trail. This ensures that the sequentially blocking trail is as a whole is selected by all agents in the block. Therefore, weak trail stability might be a more suitable solution concept for cases where contracts last longer or as a long-run prediction of outcomes. The next result is immediate so we state it without proof.

\begin{proposition}\label{theorem:stable-always-weakly-trail-stable} 
Suppose that choice functions satisfy IRC. Then any stable/chain-stable outcome is weakly trail-stable. \end{proposition}
The full substitutability assumption is not required for Proposition \ref{theorem:stable-always-weakly-trail-stable}, but, of course, the existence of stable or chain-stable outcome is not guaranteed in trading networks even under full substitutability. On the other hand, without full substitutability, trail-stable outcomes may not be weakly trail-stable as the following example shows.

\begin{example}[Trail stability does not imply weak trail stability without full substitutability]
Consider agents and contracts described in Example 1 and Figure \ref{fig:hatkom}.
Agents have the following preferences:
\begin{align*}
\succ_m&:\{ w\} \succ_m \emptyset\\
\succ_i&: \{ x\} \succ_i \emptyset\\
\succ_k&: \{ z,y\} \succ_k \emptyset\\
\succ_j&: \{w,x,z,y\}\succ_j\{w,z\}\succ_j\emptyset.
\end{align*}
and other outcomes are not acceptable. The preferences of all agents, except $j$, are fully substitutable. The empty set is a trail-stable outcome, but it is not weakly trail-stable since $\{w,z,x,y\}$ is a sequentially blocking trail when $m$ makes the first offer.
\end{example}

However, under full substitutability trail-stable outcomes are always weakly trail-stable.

\begin{proposition}\label{path-path}
Suppose that choice functions satisfy full substitutability and IRC. Then any trail-stable outcome is weakly trail-stable.
\end{proposition}

The existence of weakly trail-stable outcomes under full substitutability is therefore an immediate consequence of Theorem \ref{main} and Proposition \ref{path-path}.

\begin{corollary}\label{theorem:on-odd-cycle-contracts} 
Suppose that choice functions satisfy full substitutability and IRC. Then there exists a weakly trail-stable outcome. \end{corollary}

One may wonder whether under full substitutability trail-stable and weakly trail-stable outcomes in fact coincide. This is not the case---the converse of Proposition \ref{path-path} is false as the next example shows.
\begin{example}[Weakly trail-stable outcomes are not always trail-stable even under full substitutability] \label{ex:not-fully-weakly trail-stable}
Consider agents and contracts described in Example 1 and Figure \ref{fig:hatkom}.
Agents have the following preferences that induce fully substitutable choice functions:
\begin{align*}
\succ_m&:\{ w\} \succ_m \emptyset\\
\succ_i&: \{ x\} \succ_i \emptyset\\
\succ_k&: \{ z,y\} \succ_k \emptyset\\
\succ_j&: \{z,y\}\succ_j\{w,z\}\succ_j\{y,x\}\succ_j\emptyset.
\end{align*}
and other outcomes are not acceptable.
For outcome $A=\emptyset$, the trail $\{w,z,y,x\}$ is locally blocking
trail but not trail-blocking. Therefore, weakly trail-stable outcomes are
$\emptyset$ and $\{z,y\}$ but the only trail-stable outcome is
$\{z,y\}$. 
\end{example}

Therefore, under full substitutability, trail stability is a refinement of weak trail stability.

Figure \ref{fig:relationship} summarizes the relationships between various solution concepts in general trading networks that we have established in this paper. Stable and chain-stable outcomes may not exist even under full substitutability and they are not equivalent without the additional LAD/LAS assumption (see Example 1 in \citealp{HaKoNiOsWe:14}).

%\begin{figure}[t]
%\centering{}\scalebox{0.6}{\includegraphics{StabilitiesNonSep}}
%\end{figure}

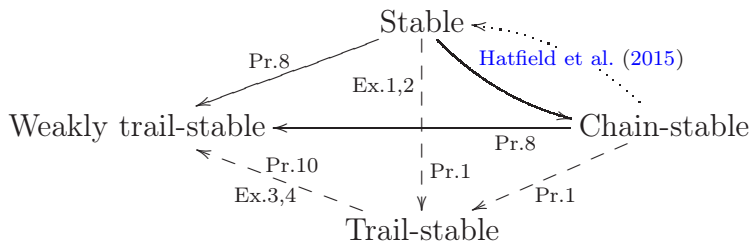
\begin{figure}[t]
\centerline{\xymatrix{
& \text{Stable} \ar[ld]^>>>>>>>>>>>>>>>>>{}_{\text{Pr.8}} \ar@{-->}[dd]^>>>>>>{\text{Pr.1}}_<<<<<<<{\text{Ex.1,2}} \ar@/_1pc/[rd] &\\
\text{Weakly trail-stable} & & \text{Chain-stable} \ar[ll]_>>>>>>>>{}^<<<<<<<<{\text{Pr.8}}\ar@{-->}[dl]^{\text{Pr.1}} \ar@{.>}@/_1pc/[ul]|>>>>>>>>>{\text{\citet{HaKoNiOsWe:14}}}\\
&\text{Trail-stable}\ar@{-->}[ul]^{\text{Ex.3,4}}_{\text{Pr.10}}&}}

\protect\caption{Relationships between solution concepts in general trading networks. Solid line: holds under IRC. Dashed line: holds under full substitutability and IRC. Dotted line: holds under full substitutability, IRC, and LAD/LAS. Arrows show which propositions establish the relationship and which examples examine the assumptions or the converse.\label{fig:relationship}} 

\end{figure}

\section{Conclusion}\label{concsec}
Stability is an appealing solution concept, but in general trading networks stable outcomes may not exist. In this paper, we introduced a new solution concept for general trading networks, called trail stability. We showed that any trading network has a trail-stable outcome when
choice functions are fully substitutable. Indeed, full substitutability is crucial for existence of trail-stable outcomes since previous maximal domain results for many-to-many matching markets apply in our case (see, for example, \citet[Theorem 6]{HaKo:12} and \citet[Theorem 2]{HaKo:11}).\footnote{When firms have quasilinear utility functions, (full) substitutability
is not necessary for competitive equilibrium and even when all agents
have complementary preferences competitive equilibrium may exist \citep{BaKl:13,Drex:13,HaKo:15,Teyt:14}.} Trail-stable outcomes have a natural lattice structure and inherit a host of properties studied in two-sided and supply-chain settings. We also considered an alternative solution concept---weak trail stability---which is implied by trail stability under full substitutability. 
Trail stability is a attractive solution concept for trading networks: in a version of our model with continuous prices, \citet*{FlJaJaTe:17} show that competitive equilibrium outcomes are trail-stable and under full substitutability essentially any trail-stable outcome can be supported by competitive equilibrium prices. 

There are at least four fruitful areas for further research. The first would be to examine the structure of weakly trail-stable outcomes. Second, one might look for weaker sufficient conditions on preferences to establish the coincidence between weakly trail-stable, trail-stable, stable outcomes (Appendix \ref{sec:coincide}). Third, it would be interesting to find more market design applications of our model \citep{morstyn2018bilateral}. Finally, it would be useful to understand how our model can be tested and estimated empirically \citep{fox2017specifying}.

\newpage
\appendix
%%%
%\iffalse
%%%
\singlespacing
\section{Relationship to previous work}\label{app:previous}
\begin{figure}[h]
\begin{tabular}{>{\centering}p{5cm}|>{\centering}p{2cm}>{\centering}p{3cm}>{\centering}p{2cm}>{\centering}p{3cm}}
 & {\small{}General networks} & {\small{}General choice\\functions} & {\small{}Existence and structure} & {\small{}New solution concepts used}
 \tabularnewline
\hline 
\hline 

{\small{}\citet{Ostr:08}} & \ding{55}\\acyclic & \ding{51} & \ding{51} & {\small{} 
Chain-stable} \tabularnewline
\hline 
{\small{}\citet{West:10}} & \ding{55}\\acyclic & \ding{51} & \ding{51} & {\small{}Group-stable,
Core} \tabularnewline
\hline 
{\small{}\citet{HaKo:12}} & \ding{55}\\acyclic & \ding{51}& \ding{51} & {\small{}Stable} \tabularnewline
\hline 
{\small{}\citet{HaKoNiOsWe:11},\\ \citet{HaKo:15}}  & \ding{51} & \ding{55}\\quasilinear (TU)& \ding{51} & {\small{}Strong group-stable} \tabularnewline
\hline 
{\small{}\citet{HaKoNiOsWe:14}} & \ding{51} & \ding{51}\\ & \ding{55}\\ & {\small{}Chain-stable} \tabularnewline
\hline 
{\small{}This paper} & \ding{51} & \ding{51} & \ding{51} & {\small{}Trail-stable\\Weakly trail-stable}\tabularnewline
\end{tabular}
\protect\caption{Relationship to previous work.\label{contribution}}
\end{figure}

\begin{figure}[h]
\small{
\begin{tabular}{>{\centering}p{5cm}|>{\centering}p{2cm}>{\centering}p{4cm}>{\centering}p{5cm}}
Paper & Theorem & Description & Generalization in this paper \tabularnewline
\hline 
\hline 
\citet{Ostr:08} & Theorem 1 & Existence of stable outcomes & Theorem \ref{main} and Proposition \ref{theorem:on-odd-cycle-contracts} \tabularnewline
\hline 
\citet{Ostr:08};\\ \citet{HaKo:12} & Theorem 2; Theorem 4 & Buyer- and seller-optimality & Proposition \ref{weaker-lattice} and Proposition \ref{fully-stable-lattice} \tabularnewline
\hline 
\citet{HaKo:12} & Theorem 8 & Rural hospitals theorem & Proposition \ref{rural} \tabularnewline
\hline 
\citet{HaKo:12}  & Theorem 10 & Strategy-proofness & Proposition \ref{Strategyproof} \tabularnewline
\hline
\citet{Ostr:08} & Theorem 3  & Firm entry & Proposition \ref{ostrovsky} \tabularnewline
\hline 
\citet{HaKo:13} & Theorem & Vacancy chain dynamics & Proposition \ref{vacancy} \tabularnewline
%\hline 
%\citet{HaKoNiOsWe:11} & Theorem 1 & Competitive equilibrium & Theorem \ref{EquilibriumTheorem}
%\tabularnewline
\end{tabular}
}
\protect\caption{Previous results generalized in this paper to a trading network setting with general choice functions. \label{generalized}}
\end{figure}
\newpage
\section{Proofs}\label{app:proofs}

We first prove Theorem \ref{main} on the existence of trail-stable outcomes. We then prove Propositions \ref{weaker-lattice} and \ref{fully-stable-lattice}---these are the most technically challenging results.

The lattice structure of trail-stable outcomes immediately gives the rural hospitals theorem and the strategyproofness result for terminal agents (Propositions \ref{rural} and \ref{Strategyproof}). Then we prove Propositions \ref{ostrovsky} and \ref{vacancy} which examine comparative statics of trail-stable outcomes. 

Finally, we prove Propositions \ref{prop:stable-is-trail-stable}, \ref{path-path}, \ref{fully-vs-trail}, and \ref{group-vs-trail} which describe the relationships between stable, trail-stable, and weakly trail-stable outcomes (proof of Proposition \ref{theorem:stable-always-weakly-trail-stable} is immediate).

Note that we sometimes refer to singleton sets of contracts $\{x\}$ as ``contract $x$'' to avoid saying ``a set containing contract $x$''.

\subsection{Proof of Theorem \ref{main}}

Consider $Y^{B}$ and $Z^{S}$, which are subsets of $X$, and represent
sets of available upstream and downstream contracts for all agents,
respectively. Define a lattice $L$ with the ground set $X\times X$
with an order $\sqsubseteq$ such that $(Y^{B},Z^{S})\sqsubseteq(Y'^{B},Z'^{S})$
if $Y^{B}\subseteq Y'^{B}$ and $Z^{S}\supseteq Z'^{S}$. 

Furthermore, define a mapping $\Phi$ as follows: 
\begin{eqnarray*}
\Phi_{B}(Y^{B},Z^{S}) & = & X\setminus R_{S}(Z^{S}|Y^{B})\\
\Phi_{S}(Y^{B},Z^{S}) & = & X\setminus R_{B}(Y^{B}|Z^{S})\\
\Phi(Y^{B},Z^{S}) & = & (\Phi_{B}(Y^{B},Z^{S}),\Phi_{S}(Y^{B},Z^{S}))
\end{eqnarray*}

where $R_{S}(Z^{S}|Y^{B})=\bigcup_{f\in F}R_{S}^{f}(Z^{S}|Y^{B})$
and $R_{B}(Y^{B}|Z^{S})=\bigcup_{f\in F}R_{B}^{f}(Y^{B}|Z^{S}) $.
Clearly, $\Phi$ is isotone \citep{Flei:03,Ostr:08,HaKo:12} on $L$.
We rely on the following well-known fixed point theorem of Tarski.

\begin{theorem}\citep{%Knas:28,
Tars:55}\label{Knaster-Tarski} Let
$L$ be a complete lattice and let $\Phi:L\rightarrow L$ be an isotone
mapping. Then the set of fixed points of $\Phi$ in $L$ is also a
complete lattice.\end{theorem}

Subsequent to the circulation of the first draft of this paper, \citet{adachi2017stable} gave an alternative proof of this Theorem \ref{main} using the $T$-operator defined by \citet{Ostr:08}.

\begin{proof}[Proof of Theorem \ref{main}]

Existence of fixed-points of $\Phi$ follows from Theorem \ref{Knaster-Tarski} since $\left(X\times X,\sqsubseteq\right)$ is a complete lattice.%
\footnote{Hence, we do not actually require the assumption of the finiteness
of contracts as long as lattice $L$ is appropriately defined. However,
we maintain this assumption for ease of comparison with previous results.%
}

We claim that every fixed point $(\dot X^{B},\dot X^{S})$ of $\Phi$
corresponds to an outcome $\dot X^{B}\cap \dot X^{S}=A$ that is {trail}-stable.
First, we show that $A$ is acceptable. We claim that if $(\dot X^{B},\dot X^{S})$ is a fixed point then $\dot X^S \cup \dot X^B=X$. To see this suppose for contradiction that there is a contract $x\notin \dot X^S \cup \dot X^B$. Then  $x \notin R_{S}(\dot X^{S}|\dot X^{B})$ therefore $x \in X\setminus R_{S}(\dot X^{S}|\dot X^{B})=\dot X^B$. So it must be that $x \in \dot X^S \cup\dot  X^B$.
This implies that $ R_{S}(\dot X^{S}|\dot X^{B})=X\setminus \dot X^B=\dot X^S \setminus A$ so $C_{S}(\dot X^{S}|\dot X^{B})=A$ and similarly  $C_{B}(\dot X^{B}|\dot X^{S})=A$. From this, we can see that $A$ is acceptable.

Second, we show that $A$ is trail-stable. This is similar to
Step 1 of the Proof of Lemma 1 in \citet{Ostr:08}. Suppose that $T=\{x_{1},\ldots,x_{M}\}$
is a locally blocking trail and assume towards a contradiction that  $T\cap A=\emptyset$.
Since we have that $x_{1}\in C_{S}^{s(x_{1})}(A\cup \{x_{1}\}|A)$, we
must have that $x_{1}\in C_{S}^{s(x_{1})}(\dot X^{S}\cup \{x_{1}\}|A)$. 
Since if  $C_{S}^{s(x_{1})}(\dot X^{S}\cup \{x_{1}\}|A) \subseteq \dot X^S$ then by IRC $C_{S}^{s(x_{1})}(\dot X^{S}\cup \{x_{1}\}|A)=A$, therefore $ C_{S}^{s(x_{1})}(A\cup \{x_{1}\}|A)=A$. 
 We also have that $x_{1}\in C_{S}^{s(x_{1})}(\dot X^{S}\cup \{x_{1}\}|\dot X^{B})$ by
CSC. If $x_{1}\in \dot X^{S}$, then $x_{1}\in \dot X^{B}=X\setminus R_{S}(\dot X^{S}|\dot X^{B})$.
But we assumed that $x_{1}\notin A$, so $x_{1}\in\dot  X^{B}$. 

Now, consider $x_{2}$. By definition of a locally blocking trail, we have that  $x_{2}\in C_{S}^{s(x_{2})}(A\cup \{x_{2}\}|A\cup \{x_{1}\})$. 
Once again by full substitutability  and IRC, we obtain that 
and $x_{2}\in C_{S}^{s(x_{2})}(\dot X^S\cup \{x_{2}\}|\dot X^B\cup \{x_{1}\})$.
If $x_{2}\in \dot X^S$, then $x_{2}\in \dot X^B=X\setminus R_{S}(\dot X^S|\dot X^B)$.
But we assumed that $x_{2}\notin A$, so $x_{2}\in \dot X^B$.
Now proceed by induction, 
we show that every $x\in T$ is in $\dot X^B$. Consider the last contract $x_M$. % is in $X_B$
Since  $x_m \in C_B^{b(x_M)} (A\cup x_{M}|A)$, using the same argument we had for $x_1$, we get that $x_M\in \dot X^S$. A contradiction.

Now we show that every trail-stable outcome corresponds to a fixed point.

Suppose $A$ is trail-stable. \ch{For every $x_i \notin A$, 
if there exists a trail $\{ x_1, x_2, \ldots, x_i \}$ such that
\begin{itemize}
\item $\{x_{1}\}$ is $(A,s(x_{1}))$-acceptable, and
\item $\{x_{m-1}, x_{m}\}$ is $(A,f_m)$-acceptable for
$f_{m}=b(x_{m-1})=s(x_{m})$ whenever $1<m\le i$, 
\end{itemize}
then let $x_i \in X_0^B$. Otherwise, let $x_i \in X_0^S$}.
Let $\dot X^B=A\cup X_0^B$ and $\dot X^S=A\cup X_0^S$.
Clearly $\dot X^S  \cup \dot X^B=X$. 

Outcome $A$ is acceptable, so $C^f(A)=A_f$ for all $f \in F$.
For every firm $f$, if $f=s(x)$ and $x\in \dot X^S\setminus A$ then  $x \notin C^f(A \cup  \{x\}) $ otherwise $x$ would be in $\dot X^B$. By SSS, we have that $C^f_S(\dot X^S |A)=A$. 
And if $f=b(y)$ and $y\in \dot X^B \setminus A$ then  $y \notin C^f(A \cup  \{y\}) $ otherwise the trail ending in $y$ would be a locally blocking trail.   By SSS, we then also have that $C^f_B(\dot X^B |A)=A$. Moreover, $\{x, y\} \nsubseteq C^f(A \cup \{x,y\})$ otherwise $x$ would be in  $\dot X^B$. %we cannot continue that path with some edge from $\dot X^S$, so  $C_S(\dot X^S |\dot X^B)=A$, $C_B(\dot X^B |\dot X^S)=A$. (???)
Putting these statements together, we have that $C_S(\dot X^S |\dot X^B)=A$ and $C_B(\dot X^B |\dot X^S)=A$.
Therefore  $R_S(\dot X^S |\dot X^B)=\dot X^S \setminus A$, $R_B(\dot X^B |\dot X^S)=\dot X^B \setminus A$, so $X \setminus R_S(\dot X^S |\dot X^B)=\dot X^B$, $X \setminus R_B(\dot X^B |\dot X^S)=\dot X^S$. \end{proof}

%\iffalse
%
%
%\subsection{Theorems \ref{LAD Theorem} and \ref{StrategyproofTheorem}}
%
%\begin{proof}[Proof of Theorem 2]
%
%Follow proof of Theorem 8 in \citet{HaKo:12} word for word.
%
%\end{proof}
%
%\begin{proof}[Proof of Theorem 3]
%
%First note that under unit demands, trail-stable outcomes are
%equivalent to weakly trail-stable outcomes. Therefore, weakly trail-stable outcome
%inherit all the properties of trail-stable outcomes from Theorem
%1. Now, to prove the remainder of the theorem, follow the proof of
%Theorem \ref{main} in \citet{HaKo:09} word for word.
%
%\end{proof}
%\fi

\subsection{Proofs of Propositions \ref{weaker-lattice} and \ref{fully-stable-lattice}}
First, we prove that fixed points of an operator $\Phi$ (defined below) form a complete sublattice (Sublattice Theorem \ref{thm:sublattice}) extending Theorems 7.3 and 7.5 from \citet{Flei:03}. Then we show two lattice properties for the terminal agents (Terminal Sublattice Theorem \ref{terminallattice} and Terminal Superiority Lemma \ref{terminalorder}). We then use these three results to prove Propositions \ref{weaker-lattice} and \ref{fully-stable-lattice}.
\subsubsection{The sublattice property of fixed points}

First note an immediate implication of the Laws of Aggregate Demand and Supply (LAD/LAS) that we have already noted in the proof of Lemma \ref{ladlasirc}. If the choice functions of firm $f$ satisfy LAD/LAS, for sets of contracts  $Y' \subseteq Y \subseteq X^B_f$, and $Z \subseteq Z' \subseteq X^S_f$ (i.e. $(Y', Z') \sqsubseteq (Y,Z)$) then $|C_{B}^{f}(Y'|Z')|-|C_{S}^{f}(Z'|Y')|\leq|C_{B}^{f}(Y|Z)|-|C_{S}^{f}(Z|Y)|$.

For every firm $f$ we define a weight function on the contracts in $X_f$, namely let $w(x)=1$  if $x \in X^B_f$ and $w(x)=-1$ if $x \in X^S_f$. So $w(C^f(Y,Z))=|C_{B}^{f}(Y|Z)|-|C_{S}^{f}(Z|Y)|$.\footnote{The weight function can be defined more generally, see \citet{Flei:03}.}
Therefore if $C^f$ satisfies LAD/LAS, then $(Y', Z') \sqsubseteq (Y,Z)$ implies $w(C^f(Y', Z')) \le w(C^f (Y,Z))$.

Let $ Y$ and $Y'$ be subsets of $X^B_f$,  $Z$ and $Z'$ are subsets of  $X^S_f$. 
We denote the complement of $Z$ in $X^S_f$ with $\overline{Z}=X^S_f\setminus Z$.
Define the operation $(Y, Z) \wsm (Y',Z')= (Y \setminus Y', \overline {Z' \setminus Z)}$. 
\ch{For a given firm $f$, 
we call a set function $R: 2^X_f \to 2^X_f$ } a \emph{w-contraction} if for every $(Y', Z') \sqsubseteq (Y,Z)$ pair,  $w(R(Y,Z) \wsm R(Y',Z')) \leq w((Y,Z) \wsm (Y',Z'))$

Let us describe some properties of this $\wsm$ operation:
\begin{lemma}\label{wsm}
For a firm $f$, let $ Y$ and $Y'$ be subsets of $X^B_f$,  $Z$ and $Z'$ are subsets of  $X^S_f$ such that  $(Y', Z') \sqsubseteq (Y,Z)$. Then the following statements hold:
\begin{enumerate}
\item
   $w((Y,Z) \wsm (Y',Z')) = w((Y,Z))-w((Y',Z'))-|X^S_f|$.
\item
For any $(A,B)$ pair, $w((A,B) \wsm (Y,Z)) \le w((A,B) \wsm (Y',Z'))$. 
\item
If $(Y,Z) \sqsubseteq (A,B)$ then  the $w((A,B) \wsm (Y,Z)) = w((A,B) \wsm (Y',Z'))$ equality  implies $(Y', Z') = (Y,Z)$.
\end{enumerate}
\end{lemma}

\begin{proof}[Proof of Lemma \ref{wsm}]
Let us tackle each statement separately:
\begin{enumerate}
\item
$w((Y,Z) \wsm (Y',Z')) = |Y\setminus Y'|-|\overline{Z'\setminus Z}|=|Y|-|Y'|-|X^S_f|+|Z'|-|Z|= w((Y,Z))-w((Y',Z'))-|X^S_f|$.
\item
Since $Y \supseteq Y'$, this implies $A\setminus Y \subseteq A \setminus Y'$,
and similarly $Z \subseteq Z'$ gives us $Z\setminus B \subseteq Z' \setminus B$, so  $\overline{Z\setminus B} \supseteq \overline {Z'\setminus B}$, therefore $w((A,B) \wsm (Y,Z)) =|A\setminus Y|- |\overline{Z \setminus B}|\le |A\setminus Y'|- |\overline{Z' \setminus B}|=w((A,B) \wsm (Y',Z'))$
\item
If $w((A,B) \wsm (Y,Z)) = w((A,B) \wsm (Y',Z'))$ then equality must hold at $|A\setminus Y|= |A\setminus Y'|$ and 
$|\overline{Z \setminus B}|=|\overline{Z' \setminus B}|$. 
Since $Y' \subseteq Y \subseteq A$ and $Z' \supseteq Z \supseteq B$, we get that $Y=Y'$ and $Z=Z'$.\qedhere\end{enumerate} \end{proof}

\begin{lemma}\label{contraction}
Suppose that the choice function of $f \in F$ satisfies full substitutability and LAD/LAS. Then the rejection function $R^f$ is  $ \sqsubseteq$-isotone and a $w$-contraction.
\end{lemma}

\begin{proof}[Proof of Lemma \ref{contraction}]
Let $Y$ and  $Y'$ be subsets of $ X^B_f$ and $Z$ and $Z'$ are be of $ X^S_f$, and moreover let
$(Y', Z') \sqsubseteq (Y,Z)$.\\
We have seen earlier that $R^f$ is  $ \sqsubseteq$-isotone, so $R^f(Y', Z') \sqsubseteq  R^f(Y,Z)$.
To prove that it is $w$-contraction, 
 $w(R^f(Y,Z) \wsm R^f(Y',Z'))+|X^S_f|=w(R^f(Y,Z)) -w( R^f(Y',Z'))= w((Y,Z) \setminus C^f(Y,Z))-w((Y', Z')\setminus C^f(Y',Z'))=w(Y,Z) -w(C^f(Y,Z))- w(Y',Z') +w( C^f(Y',Z')) \leq w(Y,Z) - w(Y',Z')=w((Y,Z) \wsm (Y', Z'))+|X^S_f|$. We used that $w(C^f(Y', Z')) \le w(C^f (Y,Z))$.
If we subtract $|X^S_f|$ from both sides, we get that\\ $w(R^f(Y,Z) \wsm R^f(Y',Z'))\leq w((Y,Z) \wsm (Y', Z'))$, so $R^f$ is indeed a $w$-contraction.
\end{proof}

We will work on the $(2^{(X,X)} , \wcup, \wcap)$ lattice. We can imagine it as a network that contains exactly two (unrelated) copies of each contract (two half-contracts), \ch{one for the buyer and one for the seller of the contract.}

Now the $C^f$ choice functions of the firms are defined over  disjoint set of contracts, so for every $Y \subseteq (X,X)$ we can define $C(Y)= \bigcup_{f\in F} C^f(Y_f) $. Similarly  $R(Y)= \bigcup_{f\in F} R^f(Y_f) $. 
\ch{On this whole network,  
we call a set function $R: 2^{(X,X)} \to 2^{(X,X)}$  a \emph{w-contraction} if for every firm $f$ the corresponding $R_f$ was a $w$-contraction. }

%$(Y', Z') \sqsubseteq (Y,Z)$ pair,  $w(R(Y,Z) \wsm R(Y',Z')) \leq %w((Y,Z) \wsm (Y',Z'))$

Let us denote the set of the starting half-contracts (seller's side) with $X^S_F= \bigcup_{f \in F} X^S_f$, and the set of ending half-contracts (buyer's side) with $X^B_F= \bigcup_{f \in F} X^B_f$. Now %%$(X,X)= X^S_F \cup X^B_F$ and 
 $|X^S_F|=|X^B_F|=|X|$.

% The other operation is the following.
Let $Y \subseteq X^B_F$ and  $Z \subseteq X^S_F$.
 The \emph{dual} of $(Y,Z)$  is what we get by switching the two parts.  We denote it with $(Y,Z)^*=(Z,Y)$. 
%Formally, for any contract $a=uv$, let $a_v \in Y \Leftrightarrow a_u \in Y^*$.

%We will work on the $(2^{(X,X)} , \wcup, \wcap)$ lattice. 
In this model let all the contracts in $X^S_F$ have weight $w=-1$ and all contracts in $X^B_F$ have weight $w=1$.

\begin{lemma}\label{sublat}
If $F: 2^{(X,X)} \to 2^{(X,X)}$ is $\sqsubseteq$-isotone and a $w$-contraction then fixed points of $F$ form a nonempty sublattice of $(2^{(X,X)}, \widetilde \cup, \widetilde \cap)$.
\end{lemma}

\begin{proof}[Proof of Lemma \ref{sublat}]
By Theorem \ref{Knaster-Tarski}, the set of fixed points is nonempty. Now let $U \subseteq (X,X)$ and $V\subseteq (X,X) $.  Assume that $F(U)=U$ and $F(V)=V$. By monotonicity, $U\wcap V=F(U) \wcap F(V)  \sqsupseteq F(U\wcap V)$ and $U\wcup V=F(U) \wcup F(V) \sqsubseteq F(U\wcap V)$. 
From the $w$-contraction property and Lemma \ref{wsm}, we have that
\[w(U\wsm (U\wcap V)) \ge w(F(U)\wsm F(U\wcap V)) \ge w(U\wsm (U\wcap V)),\] 
\[w((U\wcup V)\wsm U) \ge w(F(U \wcup V)\wsm F(U)) \ge w((U\wcup V) \wsm U),\]
hence an equality must hold throughout. %In particular $w(U\wsm (U\wcap V)) = w(U\wsm F(U\wcap V))$
Using the third part of Lemma \ref{wsm} we can see that  $(U\wcap V) =  F(U\wcap V)$ and  $(U\wcup V) =  F(U\wcup V)$ so they are also fixed points of $F$. 
\end{proof}

\begin{obs}\label{star}
Consider two sets of contracts $(Y,Z)$ and  $(Y',Z')$ , where $Y, Y' \subseteq X^B_F$ and  $Z, Z' \subseteq X^S_F$ and 
$(X,X)\setminus (Y,Z) =(X\setminus Y, X \setminus Z)$.
If $(Y', Z') \sqsubseteq (Y,Z)$, then  
%from definition $((Y,Z) \wsm (Y',Z')) = (Y \setminus Y', \overline {Z' \setminus Z)}$. 
%On the other hand,
$((X\setminus Y, X \setminus Z)^* \wsm (X\setminus Y', X \setminus Z')^*)=((X\setminus Z)\setminus (X\setminus Z'), \overline {(X\setminus Y')\setminus (X\setminus Y)})=((Z' \setminus Z), \overline {(Y \setminus Y')})=((X,X)\setminus ((Y,Z) \wsm (Y',Z'))^*$.
\end{obs}

\begin{theorem}[Sublattice Theorem]\label{thm:sublattice}
Suppose that choice functions satisfy full substitutability and LAD/LAS. Then the fixed points of 
$\Phi(Y,Z) = ( X\setminus R_{S}(Z|Y), X\setminus R_{B}(Y|Z))$ form a nonempty, complete sublattice of $(2^X \times 2^X, \widetilde \cup, \widetilde \cap)$.
\end{theorem}

\begin{proof}[Proof of Theorem \ref{thm:sublattice}]

The $\Phi(Y,Z) = ( X\setminus R_S(Z|Y), X\setminus R_B(Y|Z))$ function can be also written as  $\Phi (Y)=((X,X) \setminus R(Y,Z))^*$. Since $R$ is $\sqsubseteq$-isotone,    $\Phi$ is also $\sqsubseteq$-isotone.  We need to show that $\Phi$ is a $w$-contraction. Suppose that $(Y',Z')  \sqsubseteq (Y,Z)$. Using Observation \ref{star}, $w(\Phi (Y,Z) \wsm \Phi (Y',Z'))=w(((X,X) \setminus R(Y,Z))^* \wsm ((X,X) \setminus R(Y',Z'))^*) = w(((X,X)\setminus (R (Y,Z) \wsm R (Y',Z')))^*)= w(R (Y,Z) \wsm R (Y',Z')) \le w( (Y,Z) \wsm  (Y',Z'))$ because in  Lemma \ref{contraction} we showed that $R$ is a $w$-contraction.

Since $\Phi$ is  $\sqsubseteq$-isotone and a $w$-contraction,  Lemma \ref{sublat} gives that the fixed points of $\Phi$ form a sublattice of $(2^{(X,X)}, \widetilde \cup, \widetilde \cap)$. 
\end{proof}

\subsubsection{Lattice for the terminal agents}
The following path independence condition was introduced by \citet{aizerman1981general}. It has been deeply explored in many-to-one matching markets by \citet{echenique2015control} and in many-to-many matching markets by \citet{Flei:03} and \citet{chambers2017choice}.
\begin{lemma}(Path Independence)\label{PI}
If choice function $C^f:2^X\to 2^X$
is same-side substitutable and satisfies IRC then  $C^f(Y\cup Z)=C^f(Y\cup C^f(Z))$ holds for $Y,Z\subseteq X$.
\end{lemma}
\begin{proof}[Proof of Lemma \ref{PI}]
Since $C^f$ is same-side substitutable, $C^f(Y\cup Z) \subseteq (Y\cup C^f(Z))$. Using IRC we have that
$C^f(Y\cup Z) \subseteq (Y\cup C^f(Z)) \subseteq (Y \cup Z)$ implies that $C^f(Y\cup Z)=C^f(Y\cup C^f(Z))$.
\end{proof}

\begin{lemma}\label{partialorder}
Suppose that choice functions satisfy full substitutability and IRC. Then terminal superiority is a partial order on \ch{terminal-trail-stable} outcomes.
\end{lemma}

\begin{proof}[Proof of Lemma \ref{partialorder}]
We need to prove that $\preceq^S$ is reflexive, antisymmetric and
transitive. Assume that $A,A'$ and $A''$ are acceptable outcomes. As
$C^f(A_f\cup A_f)=C^f(A_f)=A_f$ holds for each agent (and hence for each terminal seller)
$f$, relation $\preceq^S$ is reflexive. If $A\preceq^SA'\preceq^SA$ then we
have $A_f=C^f(A_f\cup A'_f)=A'_f$ holds for any terminal agent $f$, hence
$A=A'$ and $\preceq^S$ is antisymmetric. For transitivity, assume that
$A\succeq^SA'\succeq^SA''$. Using this and Lemma \ref{PI}, we get for any
terminal agent $f$ that
\[
C^f(A_f\cup A''_f)=C^f(C^f(A_f\cup A'_f)\cup A''_f)=C^f(A_f\cup A'_f\cup
A''_f)=C^f( A_f\cup C^f(A'_f\cup A''_f))=C^f( A_f\cup A'_f)=A_f\ ,
\]
hence $A\succeq^S A''$ indeed holds. This completes the proof.
\end{proof}

\begin{theorem}[Terminal Sublattice Theorem]\label{terminallattice}
If $L$ is a nonempty complete sublattice of  $(2^X \times 2^X, \widetilde \cup, \widetilde \cap)$ then 
$L_\mathcal{T}=\{(Y_\mathcal{T}, Z_\mathcal{T}): (Y, Z) \in L\}$ is a sublattice of  $(2^\mathcal{T} \times 2^\mathcal{T}, \widetilde \cup, \widetilde \cap)$. 
\end{theorem}

\begin{proof}[Proof of Theorem \ref{terminallattice}]
For a given $(A_\mathcal{T}, B_\mathcal{T})$ there can be more than one inverse image in the original lattice. Let $(A^*, B^*)= \widetilde  \bigcup \{(Y, Z) \in L: (Y_\mathcal{T}, Z_\mathcal{T}) \sqsubseteq (A_\mathcal{T}, B_\mathcal{T}) \}  $. Since $L$ is a complete lattice with lattice operations $ \widetilde \cup$ and  $\widetilde \cap$, this means \ch{ $(A^*, B^*) \in L$ } and $(A^*_\mathcal{T}, B^*_\mathcal{T})=  (A_\mathcal{T}, B_\mathcal{T}) $. We call it the \textit{canonical inverse image} of  $(A_\mathcal{T}, B_\mathcal{T})$, since this is the $\sqsubseteq$-greatest among all inverse images.

If $(A_\mathcal{T}, B_\mathcal{T})$ and $(C_\mathcal{T}, D_\mathcal{T}) \in L_\mathcal{T}$, let us consider $(Y,Z)= (A^*, B^*) \wcap (C^*, D^*)$. 
Since  $(Y,Z)\sqsubseteq  (A^*, B^*) $ this implies $(Y_\mathcal{T},Z_\mathcal{T})\sqsubseteq (A^*_\mathcal{T}, B^*_\mathcal{T})=(A_\mathcal{T}, B_\mathcal{T})$. Similarly $(Y_\mathcal{T},Z_\mathcal{T})\sqsubseteq (C_\mathcal{T}, D_\mathcal{T})$. We want to show that $(Y_\mathcal{T},Z_\mathcal{T})$ is the greatest lower bound of  $(A_\mathcal{T}, B_\mathcal{T})$ and $(C_\mathcal{T}, D_\mathcal{T}) $ in $ L_\mathcal{T}$.
We can see that $(Y^*, Z^* ) \sqsubseteq  (A^*, B^*) $ and $(Y^*, Z^* ) \sqsubseteq  (C^*, D^*) $ because $(A^*, B^*)$ is defined by the union of a greater set. Therefore $(Y^*, Z^* )=(Y, Z )$.

Suppose there exists a $(E_\mathcal{T},F_\mathcal{T}) \in  L_\mathcal{T}$ such that $(E_\mathcal{T},F_\mathcal{T}) \sqsubseteq (A_\mathcal{T},B_\mathcal{T})$ and  $(E_\mathcal{T},F_\mathcal{T}) \sqsubseteq (C_\mathcal{T},D_\mathcal{T})$ but $(E_\mathcal{T},F_\mathcal{T})$
\ch{$ \not \sqsubseteq$}
%is not $\sqsubseteq$ than
 $ (Y_\mathcal{T},Z_\mathcal{T})$. 
Then in the original lattice  $(E^*,F^*) \sqsubseteq (A^*,B^*)$ and  $(E^*,F^*) \sqsubseteq (C^*,D^* )$ but $(E^*,F^*) $
\ch{$\not \sqsubseteq$}
%is not $\not \sqsubseteq$ than
  $(Y^*,Z^*)$. But this is impossible since  $(Y,Z)= (A^*, B^*) \wcap (C^*, D^*)$.
So we have found a unique greatest common lower bound of  $(A_\mathcal{T}, B_\mathcal{T})$ and $(C_\mathcal{T}, D_\mathcal{T})$.

\ch{ 
Similar argument can be made in order to find the lowest common upper bound of $(A_\mathcal{T}, B_\mathcal{T})$ and $(C_\mathcal{T}, D_\mathcal{T})$.
Let $(Y,Z)= (A^*, B^*) \wcup (C^*, D^*)$. 
Since  $(Y,Z)\sqsupseteq  (A^*, B^*) $ this implies $(Y_\mathcal{T},Z_\mathcal{T})\sqsupseteq (A^*_\mathcal{T}, B^*_\mathcal{T})=(A_\mathcal{T}, B_\mathcal{T})$. Similarly $(Y_\mathcal{T},Z_\mathcal{T})\sqsupseteq (C_\mathcal{T}, D_\mathcal{T})$.

Suppose there exists a $(E_\mathcal{T},F_\mathcal{T}) \in  L_\mathcal{T}$ such that $(E_\mathcal{T},F_\mathcal{T}) \sqsupseteq (A_\mathcal{T},B_\mathcal{T})$ and  $(E_\mathcal{T},F_\mathcal{T}) \sqsupseteq (C_\mathcal{T},D_\mathcal{T})$ but $(E_\mathcal{T},F_\mathcal{T})$
\ch{$ \not \sqsupseteq$}
%is not $\sqsupseteq$ than
 $ (Y_\mathcal{T},Z_\mathcal{T})$. 
Then in the original lattice  $(E^*,F^*) \sqsupseteq (A^*,B^*)$ and  $(E^*,F^*) \sqsupseteq (C^*,D^* )$ therefore $(E^*,F^*) \sqsupseteq (Y,Z)$, so  $(E^*_\mathcal{T},F^*_\mathcal{T})=(E_\mathcal{T},F_\mathcal{T}) \sqsupseteq (Y_\mathcal{T},Z_\mathcal{T})$, which is a contradiction. 

So we have found a unique lowest common upper bound  of  $(A_\mathcal{T}, B_\mathcal{T})$ and $(C_\mathcal{T}, D_\mathcal{T})$, so  $(L_\mathcal{T}, \widetilde \cup, \widetilde \cap)$ is indeed a lattice. }\end{proof}

Now we consider only the contracts sold by the terminal sellers.  
For any \ch{ $Y \subseteq X$, let  $Y_\mathcal{S}=\{x \in Y|s(x)\in \mathcal{T}$$\}$. }
% and $X^S_V=X^S \cap X_V$.

Given two trail-stable outcomes $A$ and $A'$, let us denote the canonical trail-stable pair \ch{(defined as at the end of Proof of Theorem \ref{main}) for $A$ with $\dot X^B$ and $\dot X^S$, and the canonical trail-stable pair for $A'$ with $\dot X'^B$ and $\dot X'^S$.}
\ch{
\begin{lemma}[Terminal Superiority Lemma]\label{terminalorder} Given two trail-stable outcomes $A$ and $A'$, $C^f(A_f\cup A'_f)=A_f$ 
for each teminal seller if and only if $\dot X^S_\mathcal{S} \supseteq \dot X'^S_\mathcal{S}$ and
$\dot X^B_\mathcal{S} \subseteq \dot X'^B_\mathcal{S}$ holds. A similar statement holds for terminal
buyers.

\end{lemma}
\begin{proof}[Proof of Lemma \ref{terminalorder}]
If $f$ is a terminal seller, $C^f(\dot X^S)=A_f$ and
$C^f(\dot X'^S)=A'_f$. Suppose that $\dot X^S_{\mathcal S} \supseteq \dot X'^S_{\mathcal S}$.  By IRC, $A_f \subseteq A_f\cup A'_f \subseteq \dot X^S_f$ implies that
$C^f(A_f\cup A'_f)=A_f$.

For the opposite direction, take a contract $x \in  X_f$ such that $x \notin
C^f(A'_f \cup x)$.  We use Lemma \ref{partialorder}, $A \succeq^S A'
\succeq^S x$, therefore $A \succeq^S x$, so $x \notin C^f(A_f \cup
 \{x\})$. When we define the stable pairs for $A$ and $A'$, if $x \in C^f(A'_f
\cup \{x\})$ then $x\in \dot X^B$, if $x \notin C^f(A'_f \cup \{x\})$ then $x\in
\dot X^S$. From the previous observation we can see that $\dot X^S_{\mathcal S} \supseteq
\dot X'^S_{\mathcal S}$ and $\dot X^B_{\mathcal S} \subseteq \dot X'^B_{\mathcal S}$.
The proof for terminal buyers is analogous.
\end{proof}
}

\begin{proof}[Proof of Proposition \ref{weaker-lattice}]
In the proof of Theorem \ref{main} we have seen that any fixed point
$(\dot X^B,\dot X^S)$ of isotone mapping $\Phi$ on lattice $L$ determines a trail-stable
outcome $A^X$. Moreover, each trail-stable outcome $A$ corresponds to at least
one fixed point $(\dot X^B,\dot X^S)$ of $\Phi$. From Theorem \ref{Knaster-Tarski},
it follows that fixed points of $\Phi$ form a lattice, hence there is a
$\sqsubseteq$-minimal fixed point $(\dot Y^B,\dot Y^S)$ and a $\sqsubseteq$-maximal one
$(\dot Z^B,\dot Z^S)$. We show that trail-stable outcome $A^Y$ is seller-optimal and $A^Z$
is buyer-optimal.
So assume that $A=A^X$ is a  trail-stable outcome. As $(\dot Y^B,\dot Y^S)\sqsubseteq
(\dot X^B,\dot X^S)\sqsubseteq (\dot Z^B,\dot Z^S)$, we have $\dot Y^B\subseteq \dot X^B\subseteq \dot Z^B$ and
$\dot Y^S\supseteq \dot X^S\supseteq \dot Z^S$. Lemma \ref{terminalorder} implies that
$C^f(A_f\cup A_f^Y)=A_f^Y$ and $C^f(A_f\cup A_f^Z)=A_f$ for any terminal
seller $f$ and $C_g(A_g\cup A_g^Y)=A_g$ and $C_g(A_g\cup A_g^Z)=A_g^Z$ for
any terminal buyer $g$. So, by definition $A$ is seller-superior to $A^Y$
and $A^Z$ is seller-superior to $A$.
\end{proof}

\begin{proof}[Proof of Proposition \ref{fully-stable-lattice}]
In the proof of Theorem \ref{main} we have seen that $A$ is trail-stable if and only if there is canonical trail-stable pair $(\dot X^B,\dot X^S)$ such that $(\dot X^B,\dot X^S)$ is a fixed point of
isotone mapping $\Phi$ and $A=\dot X^B\cap \dot X^S$. Moreover, \ch{if the choice functions satisfy LAD/LAS, then fixed points of
$\Phi$ form a sublattice $L$ of $(2^X \times 2^X, \widetilde \cup, \widetilde \cap)$ by Theorem \ref{thm:sublattice}. 
From Theorem \ref{terminallattice}, the projection of the above lattice to the terminals, $L_\mathcal{T}$ is also a lattice under $\sqsubseteq$
 and from Lemma \ref{terminalorder} this partial order coincides with $\preceq^S $.
Therefore, the trail-stable outcomes form a lattice under terminal-superiority. }
\end{proof}

\subsection{Proofs of Propositions \ref{rural} and \ref{Strategyproof}}
\begin{proof}[Proof of Proposition \ref{rural}]
Follow the proof of Theorem 8 in \citet{HaKo:12} word-for-word, only replacing ``stable'' with ``trail-stable''.
\end{proof}
\begin{proof}[Proof of Proposition \ref{Strategyproof}]
Follow the proof of Theorem 1 in \citet{HaKo:09} (which was pointed out by \citet{HaKo:12} for stable outcomes in supply chains).
\end{proof}

\subsection{Proofs of Propositions \ref{ostrovsky} and \ref{vacancy}}

Our proof follows \citet{Ostr:08}. \ch{First we investigate the restabilized outcome from $A$, which we play part in the proofs of both  Propositions \ref{ostrovsky} and \ref{vacancy}.
Let $A$ be an arbitrary trail-stable outcome in the original network, with a corresponding canonical trail-stable pair  $(\dot X^B,\dot  X^S)$.}
%
%In the original network, the $\sqsubseteq$-maximal fixed point of $\Phi$  is $(\dot X^B,\dot X^S)$ and  $\dot X^B \cap \dot X^S= A_{max}$.
After the new terminal seller $f'$ arrives, let $X$ be the set of all contracts in the new network, and let us define $(X^{*B},X^{*S})=(\dot X^B,\dot X^S\cup X_{f'})$. %We consider it as a contract-set-pair in the new network.
 In the following, we will use $\Phi $
%%=$\Phi(Y,Z) = ( X\setminus R_{S}(Z|Y), X\setminus R_{B}(Y|Z))$$ 
according to the choice fuctions on the new network, \ch{so $(\dot X^B,\dot X^S)$   does not need to be a fixed point of $\Phi$ anymore.}

%Let $(Z^{+B},Z^{+S})=(Z^B \cup W ,Z^S\cup X_{f'})$, where $W=\{x=f'f \in X_{f'}: x \in C^f(x \cup A_{max})\}$ i.e. the set of contracts where $f'$ is the seller, and the buyer of $x$ would like to accept $x$ alongside $A_{max}$, with possibly dropping some of its old contracts. 

 Since $X_{f'}\cap X^{*B} =\emptyset$,  
for every firm $f\neq f'$, 
$R_S^{f}(X^{*S}|X^{*B})=R_S^{f}(\dot X^{S}|\dot X^{B}) $ and 
$R_B^{f}(X^{*B}|X^{*S})=R_B^{f}(\dot X^{B}|\dot X^{S})$. 
For example, if $f$ has a conctracts with $f'$, contract $x=f'f$ was not offered for firm $f$ in $X^{*B}$ so it does not get rejected. 

For firm $f'$, $R_S^{f'}(X^{*S}|X^{*B})=X_{f'} \setminus C_S^{f'}(X_{f'} )$ and 
$R_B^{f'}(X^{*B}|X^{*S})=\emptyset$. 

Therefore $\Phi(X^{*B},X^{*S})=(\dot X^B \cup C^{f'}(X_{f'}) ,\dot X^S\cup X_{f'})$. 

So $(X^{*B},X^{*S})\sqsubseteq \Phi(X^{*B},X^{*S})$, and $\Phi$ is $\sqsubseteq$-isotone, so $ \Phi(X^{*B},X^{*S}) \sqsubseteq \Phi( \Phi(X^{*B},X^{*S}))$ and so on. The lattice of all possible subset-pairs is finite, so there is a $k$ such that $\Phi^k(X^{*B},X^{*S})=(\hat X^B, \hat X^S)$ is a fixed point. 
So $(X^{*B},X^{*S})\sqsubseteq \Phi(X^{*B},X^{*S}) \sqsubseteq\Phi^k(X^{*B},X^{*S})=(\hat X^B, \hat X^S)$. %%\sqsubseteq (X'^{B},X'^{S})$. 
 Outcome $\hat A=\hat X^{B}\cap \hat X^{S}$ is trail-stable, and this is what we call the \textit{restabilized outcome} from $A$.

\begin{proof}[Proof of Proposition \ref{ostrovsky}]
If $f'$ is a terminal seller, and we start from outcome $A_{max}$ and  the $\sqsubseteq$-maximal pair  $(\dot Z^B, \dot  Z^S)$. 
%Let  $(Z^{*B},Z^{*S})=(Z^B,Z^S\cup X_{f'})$. 
Using the previous method, 
outcome $\hat A=\hat Z^{B}\cap \hat Z^{S}$ is the restabilized outcome from $A$. 
In the new network there exists a  $\sqsubseteq$-maximal fixed point of $\Phi$, namely $(Z'^{B},Z'^{S})$, therefore 
%$(Z^{*B},Z^{*S})\sqsubseteq \Phi(Z^{*B},Z^{*S}) \sqsubseteq
%\Phi^k(Z^{*B},Z^{*S}) \sqsubseteq (Z'^{B},Z'^{S})$. 
$(\dot Z^B,\dot Z^S\cup X_{f'})=(Z^{*B},Z^{*S})\sqsubseteq (\hat Z^B, \hat Z^S)\sqsubseteq (Z'^{B},Z'^{S})$. 
The trail-stable outcome corresponding to the maximal fixed point is $A'_{max}=Z'^{B}\cap Z'^{S}$. 
We have to show that $ A'_{max} $ is better for terminal buyers and worse for terminal sellers than the original $ A_{max}$.
If $f$ is a terminal buyer, since $(Z'^{B},Z'^{S})$ is fixed point of $\Phi$ and  $(\dot Z^{B},\dot Z^{S})$ was  fixed before the new agent arrived,  $C^f(Z'^{B})=A'_{f,max}$ and $C^f(Z^{*B})=A_{f,max}$ and  $Z^{*B}\subseteq Z'^{B}$ so from $C^f(Z'^{B}) \subseteq (A_{f,max} \cup A'_{f,max}) \subseteq Z'^{B}$ by IRC we obtain $C^f(A_{f,max} \cup A'_{f,max})=A'_{f,max}$ so $ A'_{f,max}$ is preferred by terminal buyers. 

Similarly, if $f$ is a terminal seller outside $f'$, %since $(Z'^{B},Z'^{S})$ is fixed point and  $(Z'^{B},Z'^{S})$ was  fixed before,
  $C^f(Z'^{S})=A'_{f,max}$ and $C^f(Z^{*S})=A_{f,max}$ and  $Z'^{S}\subseteq Z^{*S}$ so from $C^f(Z^{*S}) \subseteq (A_{f,max} \cup A'_{f,max}) \subseteq Z^{*S}$ by IRC we obtain $C^f(A_{f,max} \cup A'_{f,max})=A_{f,max}$ so $ A_{f,max}$ is preferred by terminal buyers. If $f' $ is a terminal buyer then we can use the same proof with reversing the roles of buyers and sellers.  \end{proof}
\begin{proof}[Proof of Proposition \ref{vacancy}]
If $f'$ is a terminal seller, and 
$A$ is any trail-stable outcome in the original network, with canonical trail-stable pair  $(\dot X^B,\dot  X^S)$, then 
$(X^{*B},X^{*S})=(\dot X^B,\dot X^S\cup X_{f'}) \sqsubseteq (\hat X^B, \hat X^S)$. The restabilized outcome is $\hat A=\hat X^{B}\cap \hat X^{S}$, and similarly to the proof of  Proposition \ref{ostrovsky} one can show that initial producers
weakly prefer $A$ to $ \hat A$ and all end consumers  (other than $f'$) prefer $\hat A$ to $A$. If $f'$ is a terminal buyer,  preferences are the opposite. \end{proof}

\subsection{Proofs of Propositions \ref{prop:stable-is-trail-stable}, \ref{path-path}, \ref{fully-vs-trail}, and \ref{group-vs-trail}}

The following lemma shows that if a locally blocking trail intersects an agent several times, but he doesn't want to pick every contract in the locally blocking trail, then the agent will still select any of his upstream (downstream) contracts alongside some another one of his downstream (upstream) contracts in the locally blocking trail.

\begin{lemma}\label{shortcut}
Suppose that choice functions satisfy full substitutability and IRC.  Moreover, consider a set of contracts $Y \subset X$ and a set of contracts $\{x_1,x_2,\ldots,x_k,z_1,z_2,\ldots,z_k\}$ for agent $f$.
\begin{enumerate}
\item Assume that $Y$ is acceptable and
$x_1,x_2,\ldots,x_k\in X^B_f$ and $z_1,z_2,\ldots,z_k\in X^S_f$ such that
% none of these contracts are $(Y,f)$-acceptable. Assume further that
$\{x_i,z_i\}$ is a $(Y,f)$-essential pair for any $1\le i\le k$, but
$\{x_1,x_2,\ldots,x_k,z_1,z_2,\ldots,z_k\}$ is not $(Y,f)$-acceptable. Then
$\{x_i,z_j\}$ is a $(Y,f)$-essential pair for some $i\ne j$.
\item Now let's remove contract $x_1$. Assume that $Y$ is acceptable and
$x_2,\ldots,x_k\in X^B_f$ and $z_1,z_2,\ldots,z_k\in X^S_f$ such that $z_1$ is $(Y,f)$-acceptable, 
% none of these contracts are $(Y,f)$-acceptable. Assume further that
$\{x_i,z_i\}$ is a $(Y,f)$-essential pair for any $2\le i\le k$, but
$\{x_2,\ldots,x_k,z_1,z_2,\ldots,z_k\}$ is not $(Y,f)$-acceptable. Then
$\{x_i,z_j\}$ is a $(Y,f)$-essential pair for some $i\ne j$.
\item Now let's remove contracts $x_1$ and $z_k$. Assume that $Y$ is acceptable and
$x_2,\ldots,x_k\in X^B_f$ and $z_1,z_2,\ldots,z_{k-1}\in X^S_f$ such that $z_1$ and $x_k$ are $(Y,f)$-acceptable, 
% none of these contracts are $(Y,f)$-acceptable. Assume further that
$\{x_i,z_i\}$ is a $(Y,f)$-essential pair for any $2\le i\le k-1$, but
$\{x_2,\ldots,x_k,z_1,z_2,\ldots,z_{k-1}\}$ is not $(Y,f)$-acceptable. Then
$\{x_i,z_j\}\neq\{x_k,z_1\}$ is a $(Y,f)$-essential pair for some $i\ne j$.
\end{enumerate}
%The above statement remains true if one or both of the contracts $x_1$ and $z_k$ are
%\textit{void}. \ch{When we say that $x_1$ is void, we mean that:
%\begin{itemize}
%\item $x_1$ is empty, so
%\item the trail starts with $z_1$, and
%\item instead of $(Y,f)$-acceptability of pair $\{x_1,z_1\}$ we need $(Y,f)$-acceptability of $z_1$ and $i\ne 1$ in the
%conclusion. }
%When we say that both $x_1$ and $z_k$ are void, we mean that there is $(Y,f)$-essential pair $\{x_i,z_j\}$,  $i\ne j$ such that $\{x_i,z_j\}  \neq \{x_k,z_1\}$.
%\end{itemize}
\end{lemma}

\begin{proof}[Proof of Lemma \ref{shortcut}]
We proof each statement in turn.
\begin{enumerate}
\item Suppose, for example, that $z_j \notin C^f(Y\cup \{x_1,x_2,\ldots,x_k,z_1,z_2,\ldots,z_k\})$  for some $j$.
Then from CSC, $z_j \notin C^f(Y\cup \{x_j,z_1,z_2,\ldots,z_k\})$.
But $x_j \in C^f( Y \cup \{x_j,z_j\})$ because $\{x_j,z_j\}$ is $(Y,f)$-essential by assumption so from CSC we must have that $x_j \in C^f(Y\cup \{x_j,z_1,z_2,\ldots,z_k\})$.
 Since $\{x_j\}$ is not $(Y,f)$-acceptable, there must be a $z_l \in C^f(Y\cup \{x_j,z_1,z_2,\ldots,z_k\})$ so that $\{x_j, z_l\}$ is $(Y,f)$-acceptable for some $l \neq j$.

\item In the case that $x_1$ has been removed, suppose that $z_1 \notin C^f(Y\cup \{x_2,\ldots,x_k,z_1,z_2,\ldots,z_k\})$. Then, by CSC, we must have that $z_1 \notin C^f(Y\cup \{z_1,z_2,\ldots,z_k\})$ but this cannot hold if $\{z_1\}$ is $(Y,f)$-acceptable but none of the other $\{z_j\}$ contracts are $(Y,f)$-acceptable by IRC and SSS.
%Suppose that  $C^f(Y\cup \{z_1,z_2,\ldots,z_k\})= Y$,  that would imply that $C^f(Y\cup \{z_1\})= Y$. But  $z_1$ is  $(Y,f)$-acceptable. 
%Therefore there exists a $z_j$, $j\neq 1$ such that $z_j$ is $(Y,f)$-acceptable. Contradiction. 
Therefore, it must be that either $x_j$ or $z_j$ (for $2\leq j \leq k$) is not chosen by $C^f(Y\cup \{x_2,\ldots,x_k,z_1,z_2,\ldots,z_k\})$. Following the argument in Part 1 above, there must be an $(Y,f)$-essential pair $\{x_i,z_j\}$.
\item Repeat the argument in Part 2.\qedhere
\end{enumerate} \end{proof}
Before we prove Proposition \ref{prop:stable-is-trail-stable}, we will introduce a useful concept.
\begin{definition}
A non-empty set of contracts $Q$ is a
\textit{circuit} if its elements can be arranged in some order $(x_{1},\ldots,x_{M})$ such that
$b(x_{m})=s(x_{m+1})$ holds for all $m\in \{1,\ldots,M-1\}$ and $b(x_M)=s(x_1)$ where $M=|Q|$.
\end{definition}
\begin{proof}[Proof of Proposition \ref{prop:stable-is-trail-stable}] To show that every stable outcome is trail-stable
consider, towards a contradiction, a stable outcome $A$ which is not trail-stable. Pick a locally blocking trail $T$. For every firm involved in $T$, if $T_f \nsubseteq C^f(A\cup T_f)$, then using  Lemma \ref{shortcut} there is a \ch{upstream-downstream pairs of contracts} $x_j\in T$ and $z_l\in T$ such that $j \neq l$ and $\{x_j, z_l\}$ is  $(A,f)$-acceptable. Select only these essential pairs and remove the other contracts. Note that the contracts that remain still form a locally blocking trail. Continue removing contracts by applying Lemma \ref{shortcut} to the locally blocking trail until the remaining locally blocking trail forms circuit $Z$ such that $Z_f \subseteq C^f(A\cup Z_f)$ for every firm $f$. This circuit is therefore a blocking set. Since $T\cap A= \emptyset$ and $Z\subseteq T$, $Z\cap A= \emptyset$. Therefore $A$ is not stable, a contradiction. \end{proof}

\begin{lemma}\label{ratpair} Suppose that choice functions satisfy full substitutability and IRC. If $Y$ and $Z$ are disjoint
sets of contracts and $f$ is an agent such that $Z_f$ is
$(Y,f)$-acceptable
then for any contract $z$ of $Z^B_f$ one of the
following options hold: 
\begin{enumerate}
\item $\{z\}$ is $(Y,f)$-acceptable, or
\item there exists
some $z'\in Z^S_f$ such that $\{z,z'\}$ is a $(Y,f)$-acceptable
pair, or
\item there are $z_1,z_2,\ldots,z_k\in Z_f^S$ such that both
   $\{z,z_1,z_2,\ldots ,z_k\}$ and $\{z_i\}$ (for $1\le i\le k$) are
  $(Y,f)$-acceptable. 
  \end{enumerate}

  For $z\in Z_f^S$ an analogous statement holds.
\end{lemma}

\begin{proof}[Proof of Lemma \ref{ratpair}]

We can suppose without loss of generality that $z\in X^B_f$.

\begin{enumerate}
\item From SSS, it follows that $z \in C^f(Y_f \cup  Z_f^S\cup \{z\})$. \ch{Assume that $C^f(Y_f  \cup  Z_f^S \cup \{z\}) \cap Z^S_f  = \emptyset$. Therefore, we have
 $C^f(Y_f  \cup  Z_f^S\cup \{z\}) \subseteq (Y_f  \cup \{z\})
\subseteq (Y_f   \cup  Z_f^S\cup \{z\})$, so from IRC
 $z \in C^f(Y_f \cup \{z\})$, so $\{z\}$ is $(Y,f)$-acceptable.}
\item If $\{z\}$ is not $(Y,f)$-acceptable then there are some  contracts $\{z_1, z_2 \dots z_k\}= C^f(Y_f  \cup  Z_f^S \cup \{z\}) \cap  Z_f^S$.  If there exists an $z_i$ such that is $\{z_i\}$ is not $(Y,f)$-acceptable, then using SSS  again, we have $z_i \in C^f(Y_f \cup  \{z, z_i\}) $.
Suppose $z \notin C^f(Y_f \cup   \{z, z_i\})$, then $C^f(Y_f \cup  \{z, z_i\}) \subseteq (Y_f \cup \{z_i\}  )$, and from IRC we have $ C^f(Y_f \cup \{z,z_i\})= C^f(Y_f \cup \{z_i\} ) $. But since $\{z_i\}$ is not $(Y,f)$-acceptable this is impossible,  therefore  $\{z, z_i\} \subseteq  C^f(Y_f \cup   \{z, z_i\})$, we achieved a $(Y,f)$-essential pair. 

\item If all of $\{z_1, z_2 \dots z_k\} $ are  $(Y,f)$-acceptable.\qedhere
\end{enumerate}
\end{proof}

A consequence of Lemma
\ref{ratpair} is that trail stability is stronger than
weak trail stability under full substitutability.

\begin{proof}[Proof of Proposition \ref{path-path}]
\ch{
Consider a trail-stable outcome $A$.
Suppose that $A$ is not weakly trail-stable, i.e. there exists a sequentially blocking trail $T=\{x_1, x_2 \dots x_M\}$ for it. Without loss of generality, we may assume that (b)ii holds in Definition \ref{trailstable}. The other case when (b)i holds can be proved analogously.

We are going to find indices $1 \le i_1 < i_2 < i_3 \dots i_{l} \le k$ such that
\begin{itemize}
\item    $\{x_{i_1}\}$ is $(A,s(i_1))$-acceptable, and
\item $b(x_{i_{m-1}})=s(x_{i_{m}})=f_m$ and $\{x_{i_{m-1}}, x_{i_{m}}\}$ is a $(A,f_m)$-essential pair   for all $1<m\le l$, and
\item  $\{x_{i_{l}}\}$ is  $(A,b(i_l))$-acceptable. 
\end{itemize}
So this subset of the trail forms a locally blocking trail $T'$.

In the sequentially blocking trail $T$, choose the last contract $x_i  \in T$ such that $\{x_i\}$ is $(A,s(x_i))$-acceptable. There is at least one contract like this, since $\{x_1\}$ is  $(A,s(x_1))$-acceptable by definition. Let $i_1=i$.

Suppose we have already found $i_1 \dots i_{m}$  that satisfies our requirements. 
If $\{x_{i_m}\}$ is  $(A,b(x_{i_m}))$-acceptable, we end the trail there, and let $l=m$.
Otherwise,  from the definition of sequentially blocking trails, for $f_{m+1}=b(x_{i_m})$, the ending subsequence
 $T_{f}^{\ge m}=\{x_{m},...,x_{M}\}\cap T_{f}$ is $(A,f_{m+1})$-acceptable. 
Using Lemma \ref{ratpair}, there is a contract $x_{i_{m+1}} \in T_{f}^{\ge m} \cap X_f^S$ such that $i_{m+1}> i_m$ and  $\{x_{i_{m-1}}, x_{i_{m}}\}$ is a $(A,f_m)$-essential pair.  

This way, we constructed a locally blocking trail, therefore $A$ is not trail-stable.
 \end{proof}}

%\begin{proof}[Proof of Theorem \ref{nongroup-locallyblocking}] 
%\end{proof}

\newpage
\section{When do solution concepts coincide?}\label{sec:coincide}
As we have seen, stable, trail-stable, and weakly trail-stable outcomes typically do not coincide in general trading networks. In this section, we introduce two sufficient conditions on agents' preferences that ensure that these solution concepts coincide.

\subsection{Trail stability and weak trail stability}
Our first restriction generalizes ``flow-based'' choice functions \citep{Flei:14,FlJaScTe:17}.

\begin{definition} 
Choice functions of $f\in F$ are \textit{separable} if for any  $A,W\subseteq X$ and $y\in X^B_f \setminus A$ and $z\in X^S_f \setminus A$, whenever $A$ is $(W,f)$-acceptable, and $\{y,z\}$ is a $(W,f)$-essential pair, then $A\cup\{y,z\}$ is $(W,f)$-acceptable.
 
\end{definition}

Separable choice functions impose a kind of independence on choices of pairs of upstream and downstream contracts. It says that whenever the firm chooses $A$ alongside some set $W$ and $\{y,z\}$ alongside $W$ (but $y$ and $z$ would not be chosen separately alongside $W$ since  $\{y,z\}$ is a $(W,f)$-essential pair), then it would choose $A\cup\{y,z\}$ alongside $W$. Suppose signing $A$ and $\{y,z\}$ are decisions made by separate units of the firm. Separable choice functions say that it can delegate the joint input-output decisions to the units because its overall choices do not require any coordination between the units. One natural example of separable choice functions is the following: suppose each firms totally orders individual upstream contracts and individual downstream contracts. Whenever a firm is offered $k$ downstream and $l$ upstream contracts, it chooses the $z$ best upstream and the $z$ best downstream contracts where $z=\min(k,l)$.

We now pin down the role of separability for weakly trail-stable and trail-stable outcomes.
\begin{proposition}\label{fully-vs-trail} 
Suppose that choice functions satisfy full substitutability, separability, and IRC. Then an outcome is trail-stable if and only if it is weakly trail-stable.
\end{proposition}
\begin{proof}[Proof of Proposition \ref{fully-vs-trail}]
Proposition \ref{path-path} implies that if outcome $A$ is trail-stable
then $A$ is also weakly trail-stable. So assume that outcome $A$ is
weakly trail-stable. If $A$ is not trail-stable then there is a locally
blocking trail $T$ to $A$. The separability property of the choice functions imply that $T$ is a sequentially blocking trail, contradicting the weak trail stability of $A$. So $A$ is trail-stable.
\end{proof}
Separability is crucial for the correspondence between trail-stable and weakly trail-stable outcomes. Separability ensures that all sequentially blocking trails are locally blocking trails. Under separability all properties of trail-stable outcomes apply to weakly trail-stable outcomes.\footnote{
We also conjecture that in any trading network $X$ if choice functions of $F$ satisfy full substitutability and \textit{only} LAD/LAS then the terminal-\textit{weakly}-trail-stable contract sets form a lattice under terminal superiority, but leave this for future work.} Note that in Example \ref{ex:not-fully-weakly trail-stable}, $j$'s preferences were not separable and weak trail-stable outcome did not coincide with trail-stable outcomes.

\subsection{Trail stability and stability}

Separability is not enough to ensure that trail stable outcomes coincide with stable outcomes, since stable outcomes might not exist under full substitutability and separability (see Example \ref{trail-not-group-stable-example}). We turn to another preference restriction.

\begin{definition}
Choice functions of $f\in F$ are \textit{simple} if there exists an ``intensity'' mapping $w:X_f\to {\mathbb R}$ such that whenever $A$ is a $(W,f)$-acceptable set for some acceptable set $W$ of contracts, then for every $y \in X^B_f \cap A$ there exists  $z \in X^S_f \cap A$ such that $w(y)>w(z)$ holds.
\end{definition}

One example of choice functions which are simple are the following: if the agent is offered a set of contracts, he picks the upstream contract $y$ with the highest intensity and a downstream contract $z$ with the lowest intensity (as long as the intensity of the $y$ is greater than of $z$, otherwise he picks nothing). For example,  if the intensity mapping $w$ represents the per-unit price of the contract, then the condition says that the firm only signs a pair of contracts if the price in the downstream contract is greater than the price in the upstream contract, while picking the highest-price downstream contract and the lowest-price upstream contract.

\begin{proposition}\label{group-vs-trail}
Suppose that choice functions satisfy full substitutability, simplicity, and IRC. Then an outcome is stable if and only if it is weakly trail-stable.
\end{proposition}
Under simplicity all properties of trail-stable outcomes (including existence!) apply to stable outcomes.
\begin{proof}[Proof of Proposition \ref{group-vs-trail}]
Proposition \ref{path-path} implies that if outcome $A$ is stable then $A$ is also trail-stable. 
Assume that outcome $A$ is trail-stable, but not stable, so it has a blocking set $Z$.

Case 1: Suppose that for any $z \in Z$ contract $\{z\}$ is neither $(A, s(z))$-acceptable nor $(A, b(z))$-acceptable. Then using Lemma \ref{ratpair} we can find a circuit $Q=\{z_1, z_2, \dots z_k\} \subseteq Z$ such that
\iffalse
 $b(z_i)=s(z_{i+1})$ for every $i$ and $b(z_i)=s(z_{i+1})$ 
\fi
  $\{z_i, z_{i+1}\}$ is an $(A,b(z_i))$-essential pair for every $1 \le i \le k$ and $\{z_k, z_1\}$ is  an $(A, b(z_k))$-essential pair. 
Since every $\{z_i, z_{i+1}\}$  an $(A,b(z_i))$-acceptable set by itself, as choice functions are simple, intensity function $w$ must strictly decrease along circuit $Q$, which is impossible.\\

Case 2: Suppose that for every $z \in Z$, $\{z\}\subseteq Z$ is $A$-acceptable. Suppose that $\{z_1\}$ is $(A, s(z_1))$-acceptable. From Lemma \ref{ratpair} we can find  a trail  $\{z_2, z_3 \dots z_k\} \subseteq Z$ such that for every $z_i$, either $\{z_i, z_{i+1}\}$ is a $(A,b(z_i))$-essential pair, (therefore $w(z_i) > w(z_{i+1})$) or there are some $y_1 \dots y_l $ such that $b(y_j)=s(z_i)$ for all $1\le j\le l$ and $\{z_i, y_1 \dots y_l\}$ is $(A,b(z_i))$-acceptable. From the simplicity property there is a $y_j$ such that $w(z_i) > w(y_j)$, this $y_j$ contract will be $z_{i+1}$.
The trail terminates at the first occasion when $\{z_i\}$ is \ch{$(A,b(z_i))$-acceptable}. 

 %vitality observing property, we can find a trail  $\{z_2, z_3 \dots z_k\} \subseteq Z$ such that $w(z_i) >w(z_{i+1})$.
 Since the intensity strictly decreases, we cannot get back to a contract used earlier in the trail, so the trail must terminate. 
 \ch{Let us pick a contract $\{z_i\}$ in the trail such that it is the last one which is $(A,s(z_i))$-acceptable}, and then choose the smallest $j$ such that $j\ge i$ and $\{z_j\}$ is $(A,b(z_j))$-acceptable. 
From Lemma \ref{ratpair}, the trail from $z_i$ to $z_j$ is locally blocking, so outcome $A$ is not trail-stable. 
\end{proof}

\newpage
\section{Path stability}\label{app:path}
We can also weaken trail stability by insisting that trails only include any agent at most once so firms only have one opportunity to recontract during a deviation. A trail $T$ is a \textit{path} if all the agents $F(T)$ involved in the trail are distinct.

\begin{definition} An outcome $A\subseteq X$ is \textit{path-stable} if 
\begin{enumerate}
\item $A$ is acceptable.
\item There is no path $P=\{x_1,x_2,\ldots,x_M\}$, such that $P\cap
  A=\emptyset$  and
\begin{enumerate}
\item$\{x_{1}\}$ is $(A,f_1)$-acceptable for $f_{1}=s(x_{1})$, and
\item$\{x_{m-1}, x_{m}\}$ is $(A,f_m)$-acceptable for
$f_{m}=b(x_{m-1})=s(x_{m})$ whenever $1<m\leq M$ and
\item$\{x_{M}\}$ is $(A,f_{M+1})$-acceptable for $f_{M+1}=b(x_{M})$.
\end{enumerate}
\end{enumerate}

\end{definition}
Path-stable outcomes rule out consecutive pairwise deviations along paths. Since every path is a trail, every trail-stable outcome is path-stable. In acyclic networks every trail is also path, so path-stable, weakly trail-stable and trail-stable outcomes coincide with stable outcomes \citep{HaKo:12}. However, as the example below shows, path stability is weaker than weak trail stability (and hence weaker than trail stability) in general trading networks under full substitutability. This is intuitive because paths allows the firms to appear in the blocking set only once therefore they rule out fewer possible blocks.

\begin{example}[Path-stable outcomes are not necessarily trail-stable]
Consider agents and contracts described in Examples 1 and 2, and Figure \ref{fig:hatkom}. Agents have the following fully substitutable preferences:
\begin{align*}
\succ_m&:\{ w\} \succ_m \emptyset\\
\succ_i&: \{ x\} \succ_i \emptyset\\
\succ_k&: \{ z,y\} \succ_k \emptyset\\
\succ_j&: \{w,x,z,y\} \succ_j \{w,z\} \succ_j \{y,x\} \succ_j \{y,z\} \succ_j \emptyset
\end{align*}
The empty set is preferred to any other set of contracts.\\
Now, for outcome $\emptyset$, the trail $\{w,z,y,x\}$ is locally blocking, but there is no blocking path for $A=\emptyset$. 
Outcome $\{z,y\}$  is, however, blocked by path  $\{w,x\}$.
Therefore the trail-stable outcome is $\{w,z,y,x\}$ and the path-stable outcomes are $\emptyset$ and $\{w,z,y,x\}$.
\end{example}

\newpage
\singlespacing

\bibliographystyle{chicago}
\bibliography{bib}

\end{document}